\def\C{\mathbb C}
\def\R{\mathbb R}
\def\Z{\mathbb Z}

\def\H{\mathcal H}

\documentclass[reqno]{amsart}
\usepackage{amssymb}
\usepackage{verbatim}

\newtheorem{thm}{Theorem}[section]

\newtheorem{lem}[thm]{Lemma}

\newtheorem{example}{Example}

\begin{document}

\title[Centralizers of regular subgroups in simple Lie groups]{Centralizers of maximal regular subgroups in simple Lie groups and relative congruence classes of representations}

\author{M. Larouche}
\address{D\'epartement de Math\'ematiques et Statistique, 
         Universit\'e de Montr\'eal,
         2920 chemin de la Tour,
         Montr\'eal, Qu\'ebec, H3T 1J4, Canada}
\email{larouche@dms.umontreal.ca}

\author{F. W. Lemire}
\address{Department of Mathematics,
         University of Windsor,
         Windsor, Ontario, Canada}
\email{lemire@uwindsor.ca}

\author{J.  Patera}
\address{Centre de Recherches Math\'ematiques,
         Universit\'e de Montr\'eal,
         C.P.6128-Centre ville,
         Montr\'eal, Qu\'ebec, H3C\,3J7, Canada}
\email{patera@crm.umontreal.ca}

\date{\today}
%%%%%%%%%%%%%%%%%%%%%%%%%%%%%%%%%%%%%%%%%%%%%%%%%%%%%%%%%%%
\begin{abstract}

In the paper we present a new, uniform and comprehensive description of centralizers of the maximal regular subgroups in compact simple Lie groups of all types and ranks. The centralizer is either a direct product of finite cyclic groups, a continuous group of rank 1, or a product, not necessarily direct, of a continuous group of rank 1 with a finite cyclic group. Explicit formulas for the action of such centralizers on irreducible representations of the simple Lie algebras are given.

\end{abstract}
%%%%%%%%%%%%%%%%%%%%%%%%%%%%%%%%%%%%%%%%%%%%%%%%%%%%%%%%%%%
\maketitle
%\tableofcontents
%%%%%%%%%%%%%%%%%%%%%%%%%%%%%%%%%%%%%%%%%%%%%%%%%%%%%%%%%%%
\section{Introduction}
%%%%%%%%%%%%%%%%%%%%%%%%%%%%%%%%%%%%%%%%%%%%%%%%%%%%%%%%%%%

Let $G$ be a connected simple Lie group with corresponding Lie algebra denoted by $L$. Let $L'$ be a maximal regular semisimple Lie subalgebra of $L$ with corresponding subgroup $G'$. The goal of this paper is to study the centralizer of $G'$ in $G$ and its action on the representations of the Lie algebra $L$.  In general these centralizers are abelian subgroups of $G$. The first complete description of the continuous centralizers, whenever they exist, was given by Borel and de Siebenthal \cite{BdeS}, while the cases of discrete centralizers were first described by Dynkin and Oni\v s\v cik \cite{DO}.

In this paper we reformulate the results of \cite{BdeS} and \cite{DO} in a more accessible manner, using tools which were not available to the original authors. The existence and structure of the centralizer is made immediately visible from a decoration of the extended Dynkin-Coxeter diagram. In addition we provide explicit formulas for the actions of these centralizers on the finite-dimensional irreducible representations of $L$ and apply this information to the branching rules of $L$ with respect to $L'$. We observe in particular that the centralizer of $G'$ in $G$ is either a direct product of finite cyclic groups (in the maximal regular semisimple  case), a continuous group of rank 1 or a product, not necessarily direct, of a continuous group of rank 1 with a finite cyclic group (in the maximal regular reductive case).

The eigenvalues of these operators serve to decompose the irreducible representations of $L$ into representations of $L'$.  Projection matrices provided in \cite{MPR,LNP,LaP} transform the weights of an irreducible representation of $L$ into weights of the representations of the subalgebra. We can include, as an additional label, the eigenvalue of the action of the centralizer vector that serves to decompose the irreducible representation of $L$ into a direct sum of representations of $L'$. Note that the representation of $L'$ corresponding to a fixed eigenvalue may not be irreducible.

In physics the importance of the centralizers has been recognized for a long time. One of the best known examples occurs in the case $SU(3)\supset SU(2)\times U_1$. Here the centralizer is a continuous  $1$-parametric subgroup denoted $U_1$. The existence, structure and application of the centralizers in specific representations is not as well known. As two of the lowest examples one can point out the cyclic groups $\Z_2$ and $\Z_3$ in $Sp(2)\supset SU(2)\times SU(2)\times \Z_2$ and $G_2\supset SU(3)\times \Z_3$ respectively. One of the consequences of the presence of a centralizer $\Z_n$ is that it splits irreducible representations of the subalgebra/subgroup into $n$ equivalence classes. Undoubtedly such classes would find a physical interpretation in some cases. We call them {\it relative congruence classes} in this paper.

Discrete centralizers of maximal regular semisimple subalgebras are found in all simple Lie algebras except $A_n$\ $(1\leq n<\infty)$. In all cases they are formed as a product of up to three cyclic groups. Continuous centralizers of maximal regular reductive subalgebras appear in all simple Lie algebras, except in $G_2$, $F_4$, and $E_8$.

Note that we use Dynkin notations and numberings for roots, weights and diagrams.

%%%%%%%%%%%%%%%%%%%%%%%%%%%%%%%%%%%%%%%%%%%%%%%%%%%%%%%%%%%
\section{The center of $G$}
%%%%%%%%%%%%%%%%%%%%%%%%%%%%%%%%%%%%%%%%%%%%%%%%%%%%%%%%%%%

We start by reviewing the well-known results concerning the center of the simple Lie groups. We use the standard notation to identify the simple Lie groups $G$ and their corresponding simple Lie algebras, namely there are four infinite classes denoted $A_n$ $(n\geq 1)$, $B_n$ $(n\geq 2)$, $C_n$ $(n\geq 2)$ and $D_n$ $(n\geq 4)$ as well as five exceptional groups/algebras denoted by $E_6$, $E_7$, $E_8$, $F_4$ and $G_2$. The structure and properties of these Lie groups and their corresponding Lie algebras is encoded in their so-called decorated extended Dynkin diagrams (see Figure~\ref{fig:dynkin+marks}). The node in these diagrams labelled by $0$ denotes $\alpha_0$, the negative of the highest root of the algebra. The remaining nodes represent the simple roots $\{\alpha_1,\dots,\alpha_n\}$ of the algebra. The mark $m_k$ on the simple root $\alpha_k$ for $k=1,\dots,n$ denotes the coefficient of $\alpha_k$ in the expansion of the highest root $-\alpha_0$ in terms of the simple roots $\alpha_i$ (see Figure~\ref{fig:dynkin+marks}). The mark on $\alpha_0$, by convention, is 1. Note that the algebras $B_2$ and $C_2$ are homomorphic, and the extended Dynkin diagram of $B_2$ is the same as the one of $C_2$, with the only difference being in the numbering of the nodes : the roots $\alpha_1$ and $\alpha_2$ are interchanged.

From \cite{Kac} we know that the conjugacy classes of elements of finite order in $G$ of rank $n$ are specified in a bijective fashion by the set of all  ($n{+}1$)-tuples of relatively prime non-negative integers. To each such ($n{+}1$)-tuple $[s_0,s_1,\dots,s_n]$ with $s_i\in \Z_{\geq 0}$ we associate the point $X$ in the fundamental region of $G$ given by $$X={\frac{s_1}{M}}\omega_1+\cdots+{\frac{s_n}{M}}\omega_n$$ where $M=s_0+\sum_{i=1}^nm_is_i$ and the $\omega_i's$ denote the fundamental weights of the algebra. The order of the element of $G$ corresponding to such an $X$ is $M$.

The elements of the center $Z(G)$ of the simple Lie group $G$ are in one-to-one correspondence with the nodes of the corresponding extended diagram that carry marks equal to 1. They are in fact associated with the corners of the fundamental region of $G$. The extension node, which always has its mark equal to 1,  refers to the identity element of $G$. Explicitly, if $\{\hat\omega_i|i=1,\dots,n\}$ denotes the basis of the Cartan subalgebra $\H$ of $L$ which is dual to the base of simple roots $\{\alpha_i|i=1,\dots,n\}$ of $\H^*$ in the sense that $\alpha_i(\hat\omega_j)=\delta_{i,j}$, then the elements of the center of $G$ consist of all elements $e^{2\pi i \hat\omega_k}$ where $\alpha_k$ has mark $m_k=1$. In Table~\ref{tab:center}, for each simple Lie group admitting a non-trivial center, we list for reference the group structure as well as a generator of the center. 

For any irreducible representation of the group $G$ the central elements act as multiples of the identity. The collection of all finite-dimensional irreducible representations can then be partitioned according to the action of the central elements. Each equivalence class of irreducible representations with respect to this equivalence is called a {\it congruence class}.  The concept of congruence classes has application in the decomposition of representations such as tensor products of irreducible representations, see for example \cite{LP}.

Let us consider an irreducible finite-dimensional representation of $G$ having highest weight $\lambda=\sum_{i=1}^n m_i\omega_i$. Let $z=e^{2\pi i \hat\omega_j}$ be a non trivial element of the center $Z(G)$. Then the eigenvalue of $z$ acting on this representation is given by $e^{2\pi i \lambda (\hat{\omega}_j)}$. 

If we write 
$$
\hat\omega_j =\frac1C \sum_{i=1}^n r_i\hat\alpha_i
$$
where $\{\hat\alpha_1,\dots,\hat\alpha_n\}$ is the
basis of $\H$ dual to the basis of fundamental weights $\{\omega_1,\dots,\omega_n\}$  i.e.
$\omega_i(\hat\alpha_j)=\delta_{i,j}$\,, and where $C$ is the determinant of the Cartan matrix of $G$, we have that 
$$
\lambda (\hat{\omega_j})=\frac1C\sum_{i=1}^n r_i m_i \,.
$$

Since the eigenvalue of the central element $z$ is $e^{2\pi i \lambda(\hat{\omega_j})}$, we are really interested in the value of $\lambda (\hat{\omega_j}) \mod \Z$, which is uniquely determined by $\zeta_z := \sum_{i=1}^n r_i m_i \mod C$. The values $\zeta_z$ are listed in Table~\ref{tab:center} for each non-trivial central element of $G$. By convention, we list the value $\zeta_z$, where $z=e^{2\pi i \hat\omega_j}$, next to the $j^{th}$ node in the extended Dynkin diagram of $G$. We write 1 next to the extension node since it represents the identity of $G$.

%%%%%%%%%%%%%%%%%%%%%%%%%%%%%%%%%%%%%%%%%%%%%%%%%%%%%%%%%%%
\section{Branching rules and projection matrices}
%%%%%%%%%%%%%%%%%%%%%%%%%%%%%%%%%%%%%%%%%%%%%%%%%%%%%%%%%%%

Reduction of weight systems of irreducible finite-dimensional representations of simple Lie algebras to weight systems of representations of their maximal reductive subalgebras has been addressed several times in the literature \cite{MPR,PSan,MPS,McP}. In physics that problem is often referred to as the \textit{computation of branching rules}.

The branching rule for $L\supset L'$, where $L'$ is a maximal reductive subalgebra of $L$, is a linear transformation between Euclidean spaces $\R^n\rightarrow\R^m,$ where $n$ and $m$ are the ranks of $L$ and $L'$ respectively. This linear transformation can be expressed in the form of an $m\times n$ matrix, the \textit{projection matrix}. A suitable choice of bases allows us to obtain integer matrix elements in all the projection matrices we use here. The main advantage of the projection matrix method is the uniformity of its application as to the different algebra-subalgebra pairs, which makes it particularly amenable to computer implementation. 

The projection matrix method, used in \cite{MPR,PSan,MPS,McP}, can also be extended to compute the branching rules of orbits of Weyl groups $W(L)$ of semisimple Lie algebras $L$. An orbit of $W(L)$ is a finite set of points of $\R_n$ obtained from the action of $W(L)$ on a single point of $\R_n$. Weyl group orbits are closely related to weight systems of finite-dimensional irreducible representations of semisimple Lie algebras. Indeed a weight system consists of many orbits of the corresponding Weyl group, a specific orbit often appearing more than once. Which orbits a particular representation is comprised of is well known, and extensive tables of multiplicities of dominant weights can be found in \cite{BMP}. Considering the reduction of individual orbits rather than of entire weight systems offers some advantages, one of which is computational : while the number weights of a weight system grows without limits with the dimension of the representation, the number of points of an individual orbit is at most the order of the corresponding Weyl group. When dealing with large-scale computation for representations, one often needs to break down the problem into smaller ones for individual orbits. Orbit-orbit branching rules are computed with the projection matrix method for orbits of $W(A_n)$ in \cite{LNP} and for orbits of $W(B_n)$, $W(C_n)$ and $W(D_n)$ in \cite{LaP}.

The projection matrix $P$ for a particular pair $L\supset L'$ is calculated from one known branching rule. The classification of maximal reductive subalgebras of simple Lie algebras \cite{BdeS,Dynkin} provides the information to find that branching rule. The projection matrix is then obtained using the weight systems of the representations, by requiring that weights of $L$ be transformed by $P$ to weights of $L'$. Since any ordering of the weights is admissible, the projection matrix is not unique. However, by ordering the weights of $L$ by levels in a cigar-shaped structure and by doing the same with the weights of $L'$, the projection matrix obtained is convenient for large-scale computation, because dominant weights of $L'$ will always be in the first half of the weights found by multiplying the weights of $L$. Hence the problem is already reduced by half.

The projection matrices we will use in this paper are the ones provided by \cite{LNP} for reductions involving the Lie algebra $A_n$, by \cite{LaP} for reductions involving the Lie algebras $B_n$, $C_n$ and $D_n$, and by \cite{MPR} for the ones involving the exceptional Lie algebras.

%%%%%%%%%%%%%%%%%%%%%%%%%%%%%%%%%%%%%%%%%%%%%%%%%%%%%%%%%%%
\section{Discrete centralizers}
%%%%%%%%%%%%%%%%%%%%%%%%%%%%%%%%%%%%%%%%%%%%%%%%%%%%%%%%%%%

Let $G$ be a connected simple Lie group with its corresponding Lie algebra of rank $n$ denoted by $L$. Any maximal regular semisimple subalgebra of $L$ having rank $n$ can be realized in terms of the extended Dynkin diagram of $L$. In fact any such subalgebra $L'$, with corresponding subgroup $G'$, corresponds to the Dynkin diagram resulting from deleting one node having prime mark from the extended Dynkin diagram of $L$. Clearly such maximal regular semisimple subalgebras occur for all simple Lie algebras except $A_n\ (1\leq n<\infty)$. Since $G'$ is a maximal regular semisimple subgroup of $G$ the centralizer $C_G(G')$ of $G'$ in $G$ consists of all elements $e^{2\pi i h}$ where $h\in \H$ has the property that for all roots $\beta$ of the subalgebra $L'$ we have $e^{2\pi i\beta(h)}=1$ or equivalently $\beta(h)\in \Z$. It follows that the centralizer is a discrete abelian subgroup of the group $G$. In fact, the centralizer contains the center $Z(G)$ of the group $G$, the center $Z(G')$ of the group $G'$ as well as the group generated by the element $e^{\frac{2\pi i}{m_k}\hat\omega_k}$, where $\alpha_k$ denotes the deleted node with prime mark $m_k$ and $\hat\omega_k$ is the element of the Cartan subalgebra $\H$ of $L$ such that $\alpha_i(\hat\omega_k)=\delta_{i,k}$ for all $i=1,\dots, n$. This situation could be complicated by the fact that these three discrete groups have a non trivial intersection, but the following lemma simplifies it.

\begin{lem}
Let $G$ be a connected simple Lie group, with corresponding Lie algebra $L$, and $G'$ be a maximal regular semisimple Lie subgroup of $G$ with corresponding subalgebra $L'$. Let $\alpha_k$ denote the deleted node from the extended Dynkin diagram of $G$, having prime mark $m_k$. Then \\
\indent (i) $C_G(G')/Z(G) \cong \Z_{m_k}$ and \\
\indent (ii) $C_G(G')=Z(G')$.
\end{lem}

\begin{proof}
\textit{(i)} From the construction of $G'$ from the Dynkin diagram of $G$, it is clear that 
$$
C_G(G')=\left\langle Z(G), Z(G'), e^{\frac{2\pi i}{m_k}\hat\omega_k} \right\rangle  
$$
and that 
$$
Z(G') \subseteq \left< Z(G), e^{\frac{2\pi i}{m_k}\hat\omega_k} \right> \,.
$$
By combining those two observations, we thus find
$$
C_G(G')=\left< Z(G), e^{\frac{2\pi i}{m_k}\hat\omega_k} \right> 
$$ 
and we have the first result
$$
C_G(G')/Z(G) \cong \Z_{m_k} \,.
$$

\textit{(ii)} To prove the second result, it suffices to find an element in the center of $G'$ which is not in the center of $G$. Suppose such an element $x$ exists : 
$$
\exists x \in Z(G')\backslash Z(G) \Rightarrow \left< x+Z(G) \right> \cong \Z_{m_k} 
$$
since $\Z_{m_k}$ is a simple group and has no proper subgroup. But from \textit{(i)} we have that
$$
\Z_{m_k} \cong C_G(G')/Z(G)
$$
and thus find
$$
Z(G') = C_G(G')\,.
$$

Now, it remains to demonstrate that such an element $x$ exists in all cases. It is easy to see that in all cases except for the $D_4 \supset A_1\oplus A_1\oplus A_1\oplus A_1$ case, whenever a node having prime mark is deleted from the extended diagram of $G$, a new node in the diagram of $G'$ has its mark equal to 1. Hence a new element is added to the center of $G'$.
In the $D_4 \supset A_1\oplus A_1\oplus A_1\oplus A_1$ case, all the remaining nodes after the deletion already had their marks equal to 1 in the diagram of $G=D_4$. However, the extended node, the $\alpha_0$ node, which corresponded to the identity in $G=D_4$, becomes a non trivial element in $G'=A_1\oplus A_1\oplus A_1\oplus A_1$ and so is the element $x$ we were looking for.
\end{proof}

\subsection{Discrete centralizers and representations}

Let $\phi(G)$ be an irreducible finite-dimensional representation of $G$, acting as a set of linear transformations in $V_\phi$. An element $z=e^{ix}$ of the centralizer $C_G(G')$ of $G'$ in $G$ acts on any finite-dimensional irreducible representation $\phi_\lambda(G')$, arising during the restriction of a representation $\phi(G)$ to the subgroup:
$$
\phi(G)=\bigoplus_\lambda\phi_\lambda(G') \quad\Longleftrightarrow\quad\phi(L)=\bigoplus_\lambda\phi_\lambda(L')
$$
as a multiple of the identity matrix:
$$
\phi(z)\phi(G)\phi(z^{-1}) =\phi(z)\left(\bigoplus_\lambda\phi_\lambda(G')\right)\phi(z^{-1})
=\bigoplus_\lambda\kappa_\lambda\phi_\lambda(G')\,.
$$
If $z^N=1$, the eigenvalues $\kappa_\lambda$ are $N$-th roots of $1$.

A discrete centralizer $C_G(G')$ is a product of cyclic groups. Hence it consists of the elements of $G$ which are of finite order. For our task it suffices to describe just one element which generates each cyclic subgroup in the centralizer. More precisely, we need to determine the eigenvalues $\kappa_\lambda$ of such elements on every $\phi_\lambda(G')$.

We are interested, within each $G$-conjugacy class, by its unique element represented by a diagonal matrix in every $V_\phi$. A general method of describing diagonal representatives of conjugacy classes of elements of finite order in $G$ is found in \cite{MP}. Here we use it just for the elements of the centralizers.

Suppose $\phi(z)$ is the diagonal matrix representing the element $z=e^{ix}\in G$ in $V_\phi$. Suppose further that $V_\phi$ is decomposed into the sum of its weight subspaces:
$$
V_\phi=\sum_\mu V_{\phi}(\mu)\,.
$$
Then for any vector $v\in V_\phi(\mu)$ we have 
$$\phi(z)v=\kappa_\mu v$$ 
where
$$\kappa_\mu =e^{i\mu(x)} \,.$$

For example, if we take $z$ to be the element $e^{\frac{2\pi i}{m_k}\hat\omega_k} \in C_G(G')$, then the eigenvalue $\kappa_\mu$ can be calculated in the same fashion as in Section 2. If $\mu=\sum_{i=1}^n m_i\omega_i$, by writing $\hat\omega_k$ in terms of the $\hat\alpha_i$'s, we get :
$$
\hat\omega_k =\frac1C \sum_{i=1}^n r_i\hat\alpha_i
$$
where $C$ is the determinant of the Cartan matrix of $G$, and therefore have that 
$$
\frac{\mu (\hat\omega_k)}{m_k}=\frac{1}{m_kC}\sum_{i=1}^n r_i m_i \,.
$$
Again, we are really interested in the value of $\frac{\mu (\hat\omega_k)}{m_k} \mod \Z$, which can be given by a congruence equation of the form $\sum_{i=1}^n r_i m_i \mod m_kC$. 

\subsection{Relative congruence classes and branching rules}

The decomposition of an irreducible representation of $L$ into a sum of irreducible representations of $L'$ is known as the branching rule for the pair $L\supset L'$. In general, if we start with a finite-dimensional representation of $G$ then the elements in the centralizer of $G'$ can be used to provide partial invariants for the summands in the branching rule. The sets of weights of $L$ on which the centralizer elements take on constant values are called \textit{relative congruence classes}.

Example 2 below illustrates the use of relative congruence classes in branching rules.

\subsection{Explanation of Tables~\ref{tab:discrete} and~\ref{tab:discrete.exc}}

Tables~\ref{tab:discrete} and~\ref{tab:discrete.exc} present the structure of the centralizers and the relative congruence classes for all maximal regular semisimple subalgebras in classical and exceptional simple Lie algebras, respectively. For each such algebra-subalgebra pair $L\supset L'$, with associated groups $G\supset G'$, we give the structure of the centralizer of $G'$ in $G$, $C_G(G')$, which is always a product of cyclic groups. 

Since $C_G(G')=Z(G')$, we give the generators of the centralizer by computing the generator of the center of each simple part of the subgroup $G'$. The embedding we choose for our task is the one provided by the corresponding projection matrix, which can be found in \cite{LaP} for the classical cases and in \cite{MPR} for the exceptional ones. As it was discussed in Section 3, this particular choice offers computation efficiency. Assume that $\{\alpha_1,\dots,\alpha_n\}$ are the simple roots of the simple Lie algebra $L$ and let $-\alpha_0$ denote the highest root. A maximal regular semisimple subalgebra $L'$ of $L$ can be realized as the subalgebra with simple roots $\{\alpha_0, \alpha_1,\dots, \alpha_n\}\setminus\{\alpha_k\}$ where $\alpha_k$ is a simple root of $L$ with prime mark $m_k$. We saw in Lemma 4.1 that the centralizer of $G'$ in $G$ is generated by the center of the Lie group $G$ together with the element
$$e^{\frac{2\pi i}{m_k}\hat\omega_k}\,.$$
In order to relate this information in the context of projection matrices we note that there exists a Weyl automorphism $\sigma$ of $L$ which transforms the subalgebra $L'$ to the corresponding subalgebra $L''$ used to produce the projection matrix for the branching rule for $L\supset L''$. If $G''$ is the Lie group associated with $L''$, it follows that the centralizer of $G''$ in $G$ is generated by the center of $G$ together with 
$$e^{\frac{2\pi i}{m_k}\sigma(\hat\omega_k)}\,.$$

As explained in Subsection 4.1, the eigenvalue of the action of an element of $C_G(G')$ on a $\mu=\sum_{i=1}^n m_i\omega_i$ weight subspace is uniquely determined by the value of its exponent, which can be given in the form of a congruence equation. In Tables~\ref{tab:discrete} and~\ref{tab:discrete.exc}, for each pair $L\supset L'$, we provide a generator of the center of each simple part of $L'$ as a congruence equation. Furthermore, we give the structure of the quotient $C_G(G')/Z(G)$ and an element that generates it. More precisely, if $C_G(G')/Z(G) = \left\langle x + Z(G)\right\rangle$, we give the element $x$, again as a congruence equation. That particular equation is really the \textit{relative congruence equation}, as it provides the new partial invariants for the summands in the branching rule. Since the element associated to the deleted node $\alpha_k$ is certainly a suitable $x$, we take $x$ to be $e^{\frac{2\pi i}{m_k}\sigma(\hat\omega_k)}$, where $\sigma$ is the automorphism corresponding to the projection matrix.

Note that for all cases where the index $k$ appears in Table~\ref{tab:discrete}, for example $B_n\supset B_k\oplus D_{n-k}$, the inequality $k \geq n-k$ holds.

\begin{example}
Let us consider the case $F_4\supset A_2\oplus A_2$. In this case the simple root $\alpha_2$ with mark 3 is deleted from the extended Dynkin diagram, where the dotted node represents the extension ($\alpha_0$):

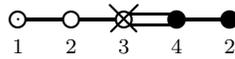
\parbox{.6\linewidth}{\setlength{\unitlength}{2pt}
\def\kr{\circle{3}}
\def\pr{\circle*{3.5}}
\thicklines
\begin{picture}(200,25)
{\footnotesize
\put(60,15){\kr}\put(59.4,14){$\cdot$}
 \put(70,15){\kr} \put(80,15){\kr}
\put(90,15){\pr} \put(100,15){\pr} \put(61.5,15){\line(1,0){7}}
\put(71.5,15){\line(1,0){7}} \put(81,16){\line(1,0){8}}
\put(81,14){\line(1,0){8}} \put(91.5,15){\line(1,0){7}}
\put(60,10){\makebox(0,0){$1$}} \put(90,10){\makebox(0,0){$4$}}
\put(70,10){\makebox(0,0){$2$}} \put(100,10){\makebox(0,0){$2$}}
\put(80,10){\makebox(0,0){$3$}} 
\put(77.6,12.7){\line(1,1){5}}\put(77.4,17.7){\line(1,-1){5}}
}
\end{picture}}

First we determine the Weyl automorphism $\sigma$, which transforms the subalgebra $(A_2 \oplus A_2)'$ having coroots $\{\hat\alpha_0, \hat\alpha_1\}$ and $\{\hat\alpha_3, \hat\alpha_4\}$ to the corresponding subalgebra $A_2\oplus A_2$ used to produce the projection matrix for the branching rule for $F_4\supset A_2\oplus A_2$. This can accomplished by noting from the projection matrix \cite{MPR}
$$\left( \begin{smallmatrix}0&0&1&1\\0&2&1&0\\1&2&1&1\\1&1&1&0 \end{smallmatrix}\right)$$
that 
$$\sigma(\hat\alpha_0)=\hat\alpha_3+\hat\alpha_4; \sigma(\hat\alpha_1)=2\hat\alpha_2+\hat\alpha_3;
\sigma(\hat\alpha_3)=\hat\alpha_1+2\hat\alpha_2+\hat\alpha_3+\hat\alpha_4;
\sigma(\hat\alpha_4)=\hat\alpha_1+\hat\alpha_2+\hat\alpha_3.$$
Therefore the coroots of $A_2\oplus A_2$ in our chosen embedding are 
$$\{\hat\alpha_3+\hat\alpha_4; 2\hat\alpha_2+\hat\alpha_3\} \,;
\{\hat\alpha_1+2\hat\alpha_2+\hat\alpha_3+\hat\alpha_4;
\hat\alpha_1+\hat\alpha_2+\hat\alpha_3\} \,.$$ 

Applying the results of Table~\ref{tab:center} for $A_n$ to the case $n=2$, we know that the eigenvalue of the action of a generator of $Z(A_2)$ on a representation of highest weight $\mu=m_1\omega_1+m_2\omega_2$ is given by the congruence equation 
$$
m_1 + 2m_2 \mod 3\,.
$$
In our situation, this information can be easily translated : the eigenvalue of the action of a generator of the center of the first $A_2$ on a $\mu=\sum_{i=1}^4 m_i\omega_i$ weight subspace is given by the congruence equation 
$$
(m_3 +m_4) + 2(2m_2+m_3) \equiv m_2+m_4 \mod 3\,.
$$
Similarly, for the second $A_2$ we find
$$
(m_1+2m_2+m_3 +m_4) + 2(m_1+m_2+m_3) \equiv m_2+m_4 \mod 3\,.
$$
Therefore, if we define $a:=m_2+m_4 \mod 3$, we have that 
$$C_{F_4}(A_2\oplus A_2)\cong \Z_3=\langle a\rangle\,.$$

Now, since the center of $F_4$ is trivial, we already know that 
$$C_{F_4}(A_2\oplus A_2)/Z(F_4)\cong \Z_3=\langle a\rangle\,.$$ 
However, let us present a uniform method to compute a generator of the quotient. We want to find the action of $e^{\frac{2\pi i}{3}\sigma(\hat\omega_2)}$. First, we compute $$\hat\omega_2=3\hat\alpha_1+6\hat\alpha_2+4\hat\alpha_3+2\hat\alpha_4\,.$$ 
Since $\hat\alpha_0=-2\hat\alpha_1-3\hat\alpha_2-2\hat\alpha_3-\hat\alpha_4$ we conclude that
$$\sigma(\hat\alpha_2)=-\hat\alpha_1-3\hat\alpha_2-2\hat\alpha_3-\hat\alpha_4\,.$$ Therefore, by substitution, we have 
$$\sigma(\hat\omega_2)=-2\hat\alpha_2-3\hat\alpha_3-2\hat\alpha_4.$$
From this we have that the action of $e^{\frac{2\pi i}{3}\sigma(\hat\omega_2)}$ on a weight $\mu=\sum_{i=1}^4 m_i\omega_i$ is given in modular form by 
$$a:m_2+m_4 \mod 3$$
and so that
$$C_{F_4}(A_2\oplus A_2)/Z(F_4)\cong \Z_3=\langle a\rangle\,.$$
 In particular we have the branching rule 
$$(1,0,0,0)\supset (0,0)(1,1)[0]+(0,2)(1,0)[1]+(2,0)(0,1)[2]+(1,1)(0,0)[0]$$
where the term in square brackets is the relative congruence class and is to be interpreted modulo 3.
\end{example}

\begin{example}
Consider $B_3\supset A_3$. We know from \cite{LaP} that the projection matrix for that case is 
$$P=\left( \begin{smallmatrix}0&1&0\\1&0&0\\0&1&1 \end{smallmatrix}\right)\,,$$
and from Table~\ref{tab:discrete} (from the line $B_n\supset D_n$, n odd, with $n=3$) that the relative congruence equation is
$$a:=2m_1+3m_3 \mod 4\,.$$
Now consider the branching rule for the irreducible $B_3$ representation with highest weight $\omega_1=(1,0,0)$. We have :
\clearpage
\begin{table}[h]
{\footnotesize
   \centering
\begin{tabular}{|c|c|c|}
\hline   
Weight & \quad $2m_1+3m_3 \mod 4$\quad & Image under P \\
\hline   
(1,0,0) & 2 & (0,1,0) \\
(-1,1,0) & 2 & (1,-1,1) \\
(0,-1,2) & 2 & (-1,0,1) \\
(0,0,0) & 0 & (0,0,0) \\
(0,1,-2) & 2 & (1,0,-1) \\
(1,-1,0) & 2 & (-1,1,-1) \\
(-1,0,0) & 2 & (0,-1,0) \\
\hline
\end{tabular}}

\bigskip
\caption{The irreducible $B_3$ representation with highest weight $(1,0,0)$.}
\end{table}
\noindent and we can conclude that the relative congruence classes split entirely the two representations of $A_3$ here. The branching rule is
$$(1,0,0)\supset (0,1,0)[2]+(0,0,0)[0]$$
where the term in square brackets is to be interpreted modulo 4.

The same exercise with the irreducible $B_3$ representation with highest weight $\omega_3=(0,0,1)$ gives us :

\begin{table}[h]
{\footnotesize
  \centering
\begin{tabular}{|c|c|c|}
\hline   
Weight & \quad $2m_1+3m_3 \mod 4$\quad & Image under P \\
\hline   
(0,0,1) & 3 & (0,0,1) \\
(0,1,-1) & 1 & (1,0,0) \\
(1,-1,1) & 1 & (-1,1,0) \\
(-1,0,1) & 1 & (0,-1,1) \\
(1,0,-1) & 3 & (0,1,-1) \\
(-1,1,-1) & 3 & (1,-1,0) \\
(0,-1,1) & 3 & (-1,0,0) \\
(0,0,-1) & 1 & (0,0,-1) \\
\hline
\end{tabular}}

\bigskip
\caption{The irreducible $B_3$ representation with highest weight $(0,0,1)$.}
\end{table}
\noindent and again we can conclude that the relative congruence classes split entirely the two representations of $A_3$. The branching rule is
$$(0,0,1)\supset (0,0,1)[3]+(1,0,0)[1]$$
where the term in square brackets is to be interpreted modulo 4.

The difference in labels in the two cases -- $0$ and $2$ for the first one, $1$ and $3$ for the second -- is caused by the fact that these two representations of $B_3$ belong to two different congruence classes : in the first case $m_3 \equiv 0 \mod 2$ whereas in the second case $m_3 \equiv 1 \mod 2$.

\end{example}

%%%%%%%%%%%%%%%%%%%%%%%%%%%%%%%%%%%%%%%%%%%%%%%%%%%%%%%%%%%
\section{Continuous centralizers}
%%%%%%%%%%%%%%%%%%%%%%%%%%%%%%%%%%%%%%%%%%%%%%%%%%%%%%%%%%%

The maximal regular reductive subalgebras of a simple Lie algebra $L$ can again be easily described in terms of the Dynkin diagram of $L$. Explicitly any such subalgebra arises as the semisimple Lie algebra having its Dynkin diagram given by deleting one node of the Dynkin diagram of $L$ having mark equal to 1 direct sum with the 1-dimensional subalgebra $\C h_0$ consisting of the intersection of the kernels of the remaining roots. If $\alpha_k$ denotes the node of mark 1 deleted from the Dynkin diagram of $L$, we observe that the centralizer $C_{G}(G')$ of $G'$ in $G$  is generated by the center $Z(G)$ of $G$, the center $Z(G')$ of $G'$  together with the rank 1 subgroup $U_1:=\left\langle e^{i\theta \hat{\omega}_k}~|~\theta\in \R \right\rangle$.  In all cases it is easily verified that the center $Z(G')$ of $G'$ is contained in the subgroup $Z(G)\times U_1$.  The centralizers of maximal regular reductive subalgebras separate into two types.  Either $e^{2\pi i\hat{\omega}_k}$ generates the center of $G$ in which case the centralizer  of $G'$ in $G$  is $U_1$ or $e^{2\pi i\hat{\omega}_k}$ generates a proper subgroup of $Z(G)$ in which case the centralizer properly contains $U_1$ -- in fact, we have $$C_{G}(G')/U_1\simeq Z(G)/\left\langle e^{2\pi i \hat{\omega}_k}\right\rangle\,.$$ 
In some cases this second type of centralizer cannot be expressed as a direct product of subgroups.

The definition of \textit{relative congruence classes} introduced in Subsection 4.2 is also true for maximal regular reductive subalgebras.

\subsection{Explanation of Table~\ref{tab:continuous}}

Table~\ref{tab:continuous} presents the structure of the centralizers and the relative congruence relations for all maximal regular reductive subalgebras in classical and exceptional simple Lie algebras, whenever such a subalgebra is present. 

For each such algebra-subalgebra pair $L\supset L'$, with associated groups $G\supset G'$, we give the structure of the centralizer of $G'$ in $G$, $C_G(G')$, which is either a continuous group of rank 1 or a product, not necessarily direct, of a continuous group of rank 1 with a finite cyclic group. Furthermore, if $L'= L''\oplus H_1$ and $G'=G''\times U_1$, we give the modular relations associated with the centers $Z(G)$ and $Z(G'')$, the structure of $H_1$ as well as the relative congruence relation provided by the centralizer $C_G(G')$.

As we did for the discrete centralizers in Section 4, the embedding of the subalgebra $L'$ we choose for our task is the one provided by the corresponding projection matrix, which can be found in \cite{MPR, LNP, LaP}. 

Assume that $\{\alpha_1,\dots,\alpha_n\}$ and $\{\omega_1,\dots,\omega_n\}$ are the simple roots and the fundamental weights of the simple Lie algebra $L$, respectively. The semisimple part $L''$ of a maximal regular reductive subalgebra $L'$ of $L$ (i.e. $L'= L''\oplus H_1$) can be realized as the subalgebra with simple roots $\{\alpha_1,\dots, \alpha_n\}\setminus\{\alpha_k\}$ where $\alpha_k$ is a simple root of $L$ with mark $m_k=1$. We have 
$$U_1:= \left\langle e^{i\theta \hat{\omega}_k}~|~\theta\in \R \right\rangle$$
or equivalently
$$H_1:= \C(\hat\omega_k)\,.$$

In order to present this information in terms of the embedding provided by the projection matrix we first note that there exists a Weyl automorphism $\sigma$ of $L$ which transforms the subalgebra $L''$ to the corresponding subalgebra $L'''$ used to produce the projection matrix for the branching rule for $L\supset L'''$. We can then determine the simple roots of $L'''$, with associated Lie group $G'''$, and write them in terms of the fundamental weights $\omega_i$'s. To compute the $H_1$ summand, one only has to find the element of the Cartan subalgebra that is in the intersection of the kernels of these weights. This in turn provides a relative congruence relation. The net effect of all these calculations is that the projection matrix of $L\supset L'''\oplus H_1$ should be written as the projection matrix of $L\supset L'''$ with an additional row at the bottom -- when a weight of an irreducible representation of $L$ is multiplied by this projection matrix, the last coordinate will yield the relative congruence value for the weight.

Now, we know the continuous rank 1 group $U_1$ is contained in the centralizer of $G'''\times U_1$ in $G$, because
$$
C_G(G'''\times U_1)=\left\langle Z(G), Z(G'''\times U_1), U_1 \right\rangle \,.
$$
And since it is easy to show that 
$$Z(G'''\times U_1)\subseteq \left\langle Z(G), U_1 \right\rangle\,,$$   
it remains to determine whether or not the center of $G$ is contained in $U_1$ to be able to finally give the structure of the centralizer. 

\begin{example} 
Let us consider the case $E_6\supset D_5\oplus H_1$. We first note from \cite{MPR} that the projection matrix for $E_6\supset D_5$ is given by
$$P=\left( \begin{smallmatrix}0&1&1&1&0&0\\ 0&0&0&0&0&1\\ 0&0&1&0&0&0\\ 0&0&0&1&1&0\\ 1&1&0&0&0&0 \end{smallmatrix}\right)\,.$$
(Note that we could use directly the projection matrix for $E_6\supset D_5\oplus H_1$, also presented in \cite{MPR}, but we choose to show here the reasoning behind that last line of the matrix.)

Since all of the roots of $E_6$ have the same length we can see that the base of simple roots for $D_5$ is given by 
$$\{\alpha_2{+}\alpha_3{+}\alpha_4, \alpha_6, \alpha_3, \alpha_4{+}\alpha_5, \alpha_1{+}\alpha_2\}\,.$$ 
Writing the roots in terms of the fundamental weights of $E_6$, we get
$$\{{-}\omega_1{+}\omega_2{+}\omega_4{-}\omega_5{-}\omega_6,{-}\omega_3{+}2\omega_6, {-}\omega_2{+}2\omega_3{-}\omega_4{-}\omega_6, {-}\omega_3{+}\omega_4{+}\omega_5, \omega_1{+}\omega_2{-}\omega_3\}\,.$$
We can determine the element of the Cartan subalgebra that is in the intersection of the kernels of these weights by simply solving a homogeneous system of linear equations, and we find :
$$H_1:=\C (\hat{\alpha}_1 -\hat{\alpha}_2+\hat{\alpha}_4-\hat{\alpha}_5)\,.$$ 

Using Table~\ref{tab:center} we find that the center of $E_6$ is generated by  
$$e^{2\pi i \hat{\omega}_1}=e^{\frac{2\pi i}{3}(4\hat{\alpha}_1 +5\hat{\alpha}_2 + 6\hat{\alpha}_3 + 4\hat{\alpha}_4 + 2\hat{\alpha}_5 + 3\hat{\alpha}_6)}\,.$$
It follows that the congruence class of an irreducible representation with highest weight $\lambda = \sum_{i=1}^6 m_i\omega_i$ is determined by the value 
$$m_1+2m_2+m_4+2m_5 \equiv m_1-m_2+m_4-m_5 \mod {3}\,.$$
Again from Table~\ref{tab:center} we have that the center of $D_5$ is generated by 
$$e^{\frac{2\pi i}{4}(2(\hat{\alpha}_2 + \hat{\alpha}_3+\hat{\alpha}_4) + 2\hat{\alpha}_3 + 3(\hat{\alpha}_4+\hat{\alpha}_5) + 5(\hat{\alpha}_1 + \hat{\alpha}_2))}$$ 
which reduces to the modular condition 
$$m_1+3m_2+m_4+3m_5 \equiv m_1-m_2+m_4-m_5 \mod {4}\,.$$ 

Finally we observe that the continuous rank 1 group 
$$U_1:=\left\langle e^{i\theta(\hat{\alpha}_1-\hat{\alpha}_2+\hat{\alpha}_4 -\hat{\alpha}_5)}~|~\theta \in \R \right\rangle$$ 
is contained in the centralizer of $D_5\times U_1$ in $E_6$ and further that the center of $E_6$ is contained in $U_1$ (take $\theta= \frac{2\pi}{3}$) and that the center of $D_5$ is contained in $U_1$ (take $\theta=\frac{2\pi}{4}$). So we naturally have that the centralizer of $D_5\times U_1$ is equal to $U_1$. For this embedding the relative congruence condition can be written as 
$$m_1-m_2+m_4-m_5\,.$$ 

Note that the effect of having the center of $E_6$ in $U_1$ is that all weights in an irreducible $E_6$ representation will yield relative congruence values, i.e. the values of $m_1{-}m_2{+}m_4{-}m_5$, that will be congruent modulo 3. In other words if the $E_6$ irreducible representation has congruence class 0 then all the relative congruence values on the weights of this representation will be congruent to 0 modulo 3 ($0,~\pm 3,~\pm 6,$~\dots).

All these calculations imply that the projection matrix of $E_6\supset D_5\oplus H_1$ should be written as the projection matrix of $E_6\supset D_5$ with an additional row at the bottom given by 
$$\left( \begin{smallmatrix}1&{-}1&0&1&{-}1&0 \end{smallmatrix}\right)\,.$$ 
When a weight of an irreducible representation of $E_6$ is multiplied by this projection matrix, the last coordinate will yield the relative congruence value for the weight.

Now consider the branching rule for the irreducible 27-dimensional $E_6$ representation with highest weight $\omega_1=(1,0,0,0,0,0)$. We have :

\begin{table}[h]
{\footnotesize
  \centering
\begin{tabular}{|c|c|c|}
\hline   
Weight & $m_1{-}m_2{+}m_4{-}m_5$ & Image under P \\
\hline   
(1,0,0,0,0,0) & 1 & (0,0,0,0,1) \\
(-1,1,0,0,0,0) & -2 & (1,0,0,0,0) \\
(0,-1,1,0,0,0) & 1 & (0,0,1,0,-1) \\
(0,0,-1,1,0,1) & 1 & (0,1,-1,1,0) \\
(0,0,0,-1,1,1) & -2 & (-1,1,0,0,0) \\
(0,0,0,1,0,-1) & 1 & (1,-1,0,1,0) \\
(0,0,0,0,-1,1) & 1 & (0,1,0,-1,0) \\
(0,0,1,-1,1,-1) & -2 & (0,-1,1,0,0) \\
(0,0,1,0,-1,-1) & 1 & (1,-1,1,-1,0) \\
(0,1,-1,0,1,0) & -2 & (0,0,-1,1,1) \\
(1,-1,0,0,1,0) & 1 & (-1,0,0,1,0) \\
(0,1,-1,1,-1,0) & 1 & (1,0,-1,0,1) \\
(-1,0,0,0,1,0) & -2 & (0,0,0,1,-1) \\
(1,-1,0,1,-1,0) & 4 & (0,0,0,0,0) \\
(0,1,0,-1,0,0) & -2 & (0,0,0,-1,1) \\
(1,-1,1,-1,0,0) & 1 & (-1,0,1,-1,0) \\
(-1,0,0,1,-1,0) & 1 & (1,0,0,0,-1) \\
(-1,0,1,-1,0,0) & -2 & (0,0,1,-1,-1) \\
(1,0,-1,0,0,1) & 1 & (-1,1,-1,0,1) \\
(-1,1,-1,0,0,1) & -2 & (0,1,-1,0,0) \\
(1,0,0,0,0,-1) & 1 & (0,-1,0,0,1) \\
(0,-1,0,0,0,1) & 1 & (-1,1,0,0,-1) \\
(-1,1,0,0,0,-1) & -2 & (1,-1,0,0,0) \\
(0,-1,1,0,0,-1) & 1 & (0,-1,1,0,-1) \\
(0,0,-1,1,0,0) & 1 & (0,0,-1,1,0) \\
(0,0,0,-1,1,0) & -2 & (-1,0,0,0,0) \\
(0,0,0,0,-1,0) & 1 & (0,0,0,-1,0) \\
\hline
\end{tabular}}

\bigskip
\caption{The irreducible $E_6$ representation with highest weight $(1,0,0,0,0,0)$.}
\end{table}
\noindent and we can conclude that the relative congruence classes split entirely the three representations of $D_5$ here. The branching rule is
$$(1,0,0,0,0,0) \supset (0,0,0,0,1)[1] + (1,0,0,0,0)[-2] + (0,0,0,0,0)[4]$$
where the term in square brackets is the relative congruence class.

\end{example}

Finally, we discuss an example where the centralizer of the maximal regular reductive subalgebra is of the second type, i.e. where $e^{2\pi i\hat{\omega}_k}$ generates a proper subalgebra of $Z(G)$ in which case the centralizer properly contains $U_1$.

\begin{example}
Let us consider the case $D_4\supset A_3\oplus H_1$. We use directly the projection matrix for $D_4\supset A_3\oplus H_1$, that can be found in \cite{LaP} :
$$P=\left( \begin{smallmatrix}1&1&0&0\\ 0&0&0&1\\ 0&1&1&0\\ 1&0&1&0 \end{smallmatrix}\right)\,.$$

Since all of the roots of $D_4$ have the same length we can see that the base of simple roots for $A_3$ is given by 
$$\{\alpha_1{+}\alpha_2, \alpha_4, \alpha_2{+}\alpha_3\}\,.$$ 
Writing the roots in terms of the fundamental weights of $D_4$, we get
$$\{\omega_1{+}\omega_2{-}\omega_3{-}\omega_4,{-}\omega_2{+}2\omega_4, {-}\omega_1{+}\omega_2{+}\omega_3{-}\omega_4\}\,.$$
We can determine the element of the Cartan subalgebra that is in the intersection of the kernels of these weights by simply solving a homogeneous system of linear equations, and we find :
$$H_1:=\C (\hat{\alpha}_1 +\hat{\alpha}_3)\,.$$ 

Using Table~\ref{tab:center} we find that the center of $D_4$ is generated by  
$$e^{2\pi i \hat{\omega}_1}=e^{\frac{2\pi i}{2}(2\hat{\alpha}_1+2\hat{\alpha}_2 +\hat{\alpha}_3+ \hat{\alpha}_4)}$$
and
$$e^{2\pi i \hat{\omega}_4}=e^{\frac{2\pi i}{2}(\hat{\alpha}_1+2\hat{\alpha}_2+\hat{\alpha}_3+ 2\hat{\alpha}_4)}\,.$$
It follows that the congruence class of an irreducible representation with highest weight $\lambda = \sum_{i=1}^4 m_i\omega_i$ is determined by the values 
$$m_3+m_4 \mod {2}$$
and
$$m_1+m_3 \mod {2}\,.$$
Again from Table~\ref{tab:center} we have that the center of $A_3$ is generated by 
$$e^{\frac{2\pi i}{4}((\hat{\alpha}_1+\hat{\alpha}_2) + 2\hat{\alpha}_4 + 3(\hat{\alpha}_2+\hat{\alpha}_3))}=e^{\frac{2\pi i}{4}(\hat{\alpha}_1+4\hat{\alpha}_2+ 3\hat{\alpha}_3+2\hat{\alpha}_4)}$$ 
which reduces to the modular condition 
$$m_1+3m_3+2m_4 \mod {4}\,.$$ 

Finally we observe that the continuous rank 1 group 
$$U_1:=\left\langle e^{i\theta(\hat{\alpha}_1+\hat{\alpha}_3)}~|~\theta \in \R \right\rangle$$ 
is contained in the centralizer of $A_3\times U_1$ in $D_4$ and further that the center of $A_3$ is contained in $Z(D_4)\times U_1$ (take $e^{2\pi i \hat{\omega}_1}$ and $\theta=\frac{2\pi}{4}$). However, the center of $D_4$ is not contained in $U_1$ : instead, we have that $Z(D_4)\cap U_1$ is a proper subgroup of $Z(D_4)$ that yields the modular condition
$$m_1+m_3 \mod {2}\,.$$

In short, we know that 
$$C_{D_4}(A_3\times U_1)=\langle Z(D_4), U_1 \rangle$$
and that
$$Z(D_4)\cong\Z_2\times \Z_2\,,$$ 
and thus using the fact that
$$Z(D_4)\cap U_1\cong \Z_2\,,$$
we find that
$$C_{D_4}(A_3\times U_1)\cong U_1\times\Z_2\,.$$

For this embedding the relative congruence condition can be written as 
$$m_1+m_3\,,$$
which coincides with the last line of $P$.

Note that the effect of not having the center of $D_4$ contained in $U_1$ is that knowing the relative congruence value does not tell us which congruence class we are dealing with. For example the modules with highest weights $(1,0,0,0)$ and $(0,0,1,0)$ will both have odd relative congruence labels but these two modules are in different congruence classes -- $m_3 +m_4 \equiv 0 \mod 2$ for $(1,0,0,0)$, and $m_3 +m_4 \equiv 1 \mod 2$ for $(0,0,1,0)$ -- and hence their weight spaces must be distinguished by the action of the whole centralizer. However when trying to reduce a representation of $D_4$, the only information that can help splitting the $A_3$ representations is the relative congruence condition.

Let us consider the branching rule for the irreducible 8-dimensional $D_4$ representation with highest weight $\omega_1=(1,0,0,0)$. We have :

\begin{table}[h]
{\footnotesize
  \centering
\begin{tabular}{|c|c|c|}
\hline   
Weight & $m_1{+}m_3$ & Image under P \\
\hline   
(1,0,0,0) & 1 & (1,0,0)[1] \\
(-1,1,0,0) & -1 & (0,0,1)[-1] \\
(0,-1,1,1) & 1 & (-1,1,0)[1] \\
(0,0,-1,1) & -1 & (0,1,-1)[-1] \\
(0,0,1,-1) & 1 & (0,-1,1)[1] \\
(0,1,-1,-1) & -1 & (1,-1,0)[-1] \\
(1,-1,0,0) & 1 & (0,0,-1)[1] \\
(-1,0,0,0) & -1 & (-1,0,0)[-1] \\
\hline
\end{tabular}}

\bigskip
\caption{The irreducible $D_4$ representation with highest weight $(1,0,0,0)$.}
\end{table}
\noindent and we can conclude that the relative congruence classes split entirely the two representations of $A_3$ here. The branching rule is
$$(1,0,0,0)\supset (1,0,0)[1]+(0,0,1)[-1]$$
where the term in square brackets is the relative congruence class (and the value of the $H_1$ term).

Finally, the branching rule for the irreducible 28-dimensional $D_4$ representation with highest weight $\omega_2=(0,1,0,0)$ is
$$(0,1,0,0)\supset (1,0,1)[0]+(0,1,0)[2]+(0,1,0)[-2]+(0,0,0)[0]\,.$$
We have here an example where the representation of $A_3$ corresponding to a fixed eigenvalue is not be irreducible : the $A_3$ representations with highest weights $(1,0,1)$ and $(0,0,0)$ share the same relative congruence label $[0]$.

\end{example}

%\clearpage
%%%%%%%%%%%%%%%%%%%%%%%%%%%%%%%%%%%%%%%%%%%%%%%%%%%%%%%%%%%
\centerline{\bf Acknowledgements}
\medskip

This work was supported in part by the Natural Sciences and Engineering Research Council of Canada and by the MIND Research Institute of Santa Ana, California. M. L. is grateful for the support she receives from the Alexander Graham Bell Scholarship. We would like to thank Iryna Kashuba, who was at the beginning of this work several years ago.
%%%%%%%%%%%%%%%%%%%%%%%%%%%%%%%%%%%%%%%%%%%%%%%%%%%%%%%%%%%

%%%%%%%%%%%%%%%%%%%%%%%%%%%%%%%%%%%%%%
%\bibliographystyle{iopart-num}
%\bibliography{mabibli}

\clearpage

%%%%%%%%%%%%%%%%%%%%%%%%%%%%%%%%%%%%%%%%%%%%%%%%%%%%%%%%%%%%%%%%%%%%%%%%%%%%%%

%%%%%%%%%%%%%%%%%%%%%%%%%%%%%%%%%%%%%%%%%%%%%%%%%%%%%%%%%%%%%%%%%%%%%%%%%%%%%%
\begin{figure}
\parbox{.6\linewidth}{\setlength{\unitlength}{2pt}
\def\kr{\circle{3}}
\def\pr{\circle*{3.5}}
\thicklines
\begin{picture}(200,290)

%{\bf \put(85,285){\makebox(0,0){Table 1: The Dynkin Numbering of the extended diagrams with marks.}}}

{\bf \put(20,275){\makebox(0,0){Dynkin Numbering}}
\put(100,275){\makebox(0,0){Marks}}}
%% A_n-ext
{\footnotesize \put(-25,255){\makebox(0,0){\large{$\mathbf{A_n}$}}}
\put(0,250){\circle{3}} 
\put(10,250){\circle{3}}
\put(30,250){\circle{3}} 
\put(40,250){\circle{3}}
\put(1.5,250){\line(1,0){7}} 
\put(11.5,250){\line(1,0){5}}
\multiput(18,250)(2,0){3}{\circle{0,1}}
\put(23.5,250){\line(1,0){5}} 
\put(31.5,250){\line(1,0){7}}
\put(20,260){\circle{3}} 
\put(1.5,251){\line(2,1){17}}
\put(38.5,251){\line(-2,1){17}} 
\put(0,245){\makebox(0,0){$1$}}
\put(30,245){\makebox(0,0){$n{-}1$}}
\put(10,245){\makebox(0,0){$2$}} 
\put(40,244.7){\makebox(0,0){$n$}}
\put(20,265){\makebox(0,0){$0$}}

\put(80,250){\circle{3}} 
\put(90,250){\circle{3}}
\put(110,250){\circle{3}} 
\put(120,250){\circle{3}}
\put(81.5,250){\line(1,0){7}} 
\put(91.5,250){\line(1,0){5}}
\multiput(98,250)(2,0){3}{\circle{0,1}}
\put(103.5,250){\line(1,0){5}} 
\put(111.5,250){\line(1,0){7}}
\put(100,260){\circle{3}} 
\put(81.5,251){\line(2,1){17}}
\put(118.5,251){\line(-2,1){17}} 
\put(80,245){\makebox(0,0){$1$}}
\put(110,245){\makebox(0,0){$1$}}
\put(90,245){\makebox(0,0){$1$}}
\put(120,245){\makebox(0,0){$1$}}
\put(100,265){\makebox(0,0){$1$}}

%% B_n
\put(-25,225){\makebox(0,0){{\large$\mathbf{B_n}$}}}
\put(0,230){\kr} 
\put(0,220){\kr} 
\put(10,225){\kr}
\put(30,225){\kr} 
\put(40,225){\pr} 
\put(8.5,226){\line(-2,1){7}}
\put(8.5,224){\line(-2,-1){7}} 
\put(11.5,225){\line(1,0){5}}
\put(23.5,225){\line(1,0){5}} 
\put(31,226){\line(1,0){8}}
\put(31,224){\line(1,0){8}}
\multiput(18,225)(2,0){3}{\circle{0,1}}
\put(0,235){\makebox(0,0){$0$}} 
\put(0,215){\makebox(0,0){$1$}}
\put(30,220){\makebox(0,0){$n{-}1$}}
\put(10,220){\makebox(0,0){$2$}} 
\put(40,219.7){\makebox(0,0){$n$}}

\put(80,230){\kr} 
\put(80,220){\kr} 
\put(90,225){\kr}
\put(110,225){\kr} 
\put(120,225){\pr}
\put(88.5,226){\line(-2,1){7}} 
\put(88.5,224){\line(-2,-1){7}}
\put(91.5,225){\line(1,0){5}} 
\put(103.5,225){\line(1,0){5}}
\put(111,226){\line(1,0){8}} 
\put(111,224){\line(1,0){8}}
\multiput(98,225)(2,0){3}{\circle{0,1}}
\put(80,235){\makebox(0,0){$1$}}
\put(80,215){\makebox(0,0){$1$}}
\put(110,220){\makebox(0,0){$2$}}
\put(90,220){\makebox(0,0){$2$}}
\put(120,220){\makebox(0,0){$2$}}

%% C_n
\put(-25,200){\makebox(0,0){{\large $\mathbf{C_n}$}}}
\put(0,200){\kr} 
\put(10,200){\pr} 
\put(30,200){\pr}
\put(40,200){\kr} 
\put(1.0,201){\line(1,0){8}}
\put(1.0,199){\line(1,0){8}} 
\put(11.5,200){\line(1,0){5}}
\put(23.5,200){\line(1,0){5}} 
\put(31,201){\line(1,0){7.8}}
\put(31,199){\line(1,0){7.9}}
\multiput(18,200)(2,0){3}{\circle{0,1}}
\put(0,195){\makebox(0,0){$0$}}
\put(30,195){\makebox(0,0){$n{-}1$}}
\put(10,195){\makebox(0,0){$1$}} 
\put(40,194.7){\makebox(0,0){$n$}}

\put(80,200){\kr} 
\put(90,200){\pr} 
\put(110,200){\pr}
\put(120,200){\kr} 
\put(81.0,201){\line(1,0){8}}
\put(81.0,199){\line(1,0){8}} 
\put(91.5,200){\line(1,0){5}}
\put(103.5,200){\line(1,0){5}} 
\put(111,201){\line(1,0){7.8}}
\put(111,199){\line(1,0){7.9}}
\multiput(98,200)(2,0){3}{\circle{0,1}}
\put(80,195){\makebox(0,0){$1$}}
\put(110,195){\makebox(0,0){$2$}}
\put(90,195){\makebox(0,0){$2$}}
\put(120,195){\makebox(0,0){$1$}}

%% D_n
\put(-25,170){\makebox(0,0){{\large $\mathbf{D_n}$}}}
\put(0,175){\kr} 
\put(0,165){\kr} 
\put(10,170){\kr}
\put(30,170){\kr} 
\put(40,175){\kr} 
\put(40,165){\kr}
\put(8.5,171){\line(-2,1){7}} 
\put(8.5,169){\line(-2,-1){7}}
\put(11.5,170){\line(1,0){5}} 
\put(23.5,170){\line(1,0){5}}
\put(31.5,171){\line(2,1){7}} 
\put(31.5,169){\line(2,-1){7}}
\multiput(18,170)(2,0){3}{\circle{0,1}}
\put(0,180){\makebox(0,0){$0$}} 
\put(0,160){\makebox(0,0){$1$}}
\put(29.8,165){\makebox(0,0){$n{-}2$}}
\put(10,165){\makebox(0,0){$2$}} 
\put(40,160){\makebox(0,0){$n$}}
\put(40,180){\makebox(0,0){$n{-}1$}}

\put(80,175){\kr} 
\put(80,165){\kr} 
\put(90,170){\kr}
\put(110,170){\kr} 
\put(120,175){\kr} 
\put(120,165){\kr}
\put(88.5,171){\line(-2,1){7}} 
\put(88.5,169){\line(-2,-1){7}}
\put(91.5,170){\line(1,0){5}} 
\put(103.5,170){\line(1,0){5}}
\put(111.5,171){\line(2,1){7}} 
\put(111.5,169){\line(2,-1){7}}
\multiput(98,170)(2,0){3}{\circle{0,1}}
\put(80,180){\makebox(0,0){$1$}}
\put(80,160){\makebox(0,0){$1$}}
\put(110,165){\makebox(0,0){$2$}}
\put(90,165){\makebox(0,0){$2$}}
\put(120,160){\makebox(0,0){$1$}}
\put(120,180){\makebox(0,0){$1$}}

%% E_6
\put(-25,135){\makebox(0,0){{\large $\mathbf{E_6}$}}}
\put(0,135){\kr} 
\put(10,135){\kr}  
\put(20,135){\kr}
\put(30,135){\kr} 
\put(40,135){\kr}  
\put(20,143){\kr}
\put(20,151){\kr} 
\put(1.5,135){\line(1,0){7}}
\put(11.5,135){\line(1,0){7}} 
\put(21.5,135){\line(1,0){7}}
\put(31.5,135){\line(1,0){7}} 
\put(20,136.5){\line(0,1){5}}
\put(20,144.5){\line(0,1){5}} 
\put(0,130){\makebox(0,0){$1$}}
\put(10,130){\makebox(0,0){$2$}} 
\put(20,130){\makebox(0,0){$3$}}
\put(30,130){\makebox(0,0){$4$}} 
\put(40,130){\makebox(0,0){$5$}}
\put(24,144){\makebox(0,0){$6$}} 
\put(24,152){\makebox(0,0){$0$}}

\put(80,135){\kr} 
\put(90,135){\kr} 
\put(100,135){\kr}
\put(110,135){\kr} 
\put(120,135){\kr} 
\put(100,143){\kr}
\put(100,151){\kr} 
\put(81.5,135){\line(1,0){7}}
\put(91.5,135){\line(1,0){7}} 
\put(101.5,135){\line(1,0){7}}
\put(111.5,135){\line(1,0){7}} 
\put(100,136.5){\line(0,1){5}}
\put(100,144.5){\line(0,1){5}} 
\put(80,130){\makebox(0,0){$1$}}
\put(90,130){\makebox(0,0){$2$}}
\put(100,130){\makebox(0,0){$3$}}
\put(110,130){\makebox(0,0){$2$}}
\put(120,130){\makebox(0,0){$1$}}
\put(104,144){\makebox(0,0){$2$}}
\put(104,152){\makebox(0,0){$1$}}

%% E_7
\put(-25,105){\makebox(0,0){{\large $\mathbf{E_7}$}}}
\put(0,105){\kr} 
\put(10,105){\kr}  
\put(20,105){\kr}
\put(30,105){\kr} 
\put(40,105){\kr}  
\put(50,105){\kr}
\put(60,105){\kr} 
\put(30,113){\kr} 
\put(1.5,105){\line(1,0){7}}
\put(11.5,105){\line(1,0){7}} 
\put(21.5,105){\line(1,0){7}}
\put(31.5,105){\line(1,0){7}} 
\put(41.5,105){\line(1,0){7}}
\put(30,106.5){\line(0,1){5}} 
\put(51.5,105){\line(1,0){7}}
\put(0,100){\makebox(0,0){$0$}} 
\put(10,100){\makebox(0,0){$1$}}
\put(20,100){\makebox(0,0){$2$}} 
\put(30,100){\makebox(0,0){$3$}}
\put(40,100){\makebox(0,0){$4$}} 
\put(60,100){\makebox(0,0){$6$}}
\put(34,114){\makebox(0,0){$7$}} 
\put(50,100){\makebox(0,0){$5$}}

\put(80,105){\kr} 
\put(90,105){\kr} 
\put(100,105){\kr}
\put(110,105){\kr} 
\put(120,105){\kr} 
\put(130,105){\kr}
\put(140,105){\kr} 
\put(110,113){\kr}
\put(81.5,105){\line(1,0){7}} 
\put(91.5,105){\line(1,0){7}}
\put(101.5,105){\line(1,0){7}} 
\put(111.5,105){\line(1,0){7}}
\put(121.5,105){\line(1,0){7}} 
\put(110,106.5){\line(0,1){5}}
\put(131.5,105){\line(1,0){7}} 
\put(80,100){\makebox(0,0){$1$}}
\put(90,100){\makebox(0,0){$2$}}
\put(100,100){\makebox(0,0){$3$}}
\put(110,100){\makebox(0,0){$4$}}
\put(120,100){\makebox(0,0){$3$}}
\put(140,100){\makebox(0,0){$1$}}
\put(114,114){\makebox(0,0){$2$}}
\put(130,100){\makebox(0,0){$2$}}

%% E_8
\put(-25,75){\makebox(0,0){{\large $\mathbf{E_8}$}}}
\put(-10,75){\kr} 
\put(0,75){\kr} 
\put(10,75){\kr}
\put(20,75){\kr} 
\put(30,75){\kr} 
\put(40,75){\kr}
\put(50,75){\kr} 
\put(60,75){\kr} 
\put(40,83){\kr}
\put(-8.5,75){\line(1,0){7}} 
\put(1.5,75){\line(1,0){7}}
\put(11.5,75){\line(1,0){7}} 
\put(21.5,75){\line(1,0){7}}
\put(31.5,75){\line(1,0){7}} 
\put(41.5,75){\line(1,0){7}}
\put(51.5,75){\line(1,0){7}} 
\put(40,76.5){\line(0,1){5}}
\put(0,70){\makebox(0,0){$1$}} 
\put(10,70){\makebox(0,0){$2$}}
\put(20,70){\makebox(0,0){$3$}} 
\put(30,70){\makebox(0,0){$4$}}
\put(40,70){\makebox(0,0){$5$}} 
\put(60,70){\makebox(0,0){$7$}}
\put(44,84){\makebox(0,0){$8$}} 
\put(50,70){\makebox(0,0){$6$}}
\put(-10,70){\makebox(0,0){$0$}}

\put(75,75){\kr} 
\put(85,75){\kr} 
\put(95,75){\kr}
\put(105,75){\kr} 
\put(115,75){\kr} 
\put(125,75){\kr}
\put(135,75){\kr} 
\put(145,75){\kr} 
\put(125,83){\kr}
\put(76.5,75){\line(1,0){7}} 
\put(86.5,75){\line(1,0){7}}
\put(96.5,75){\line(1,0){7}} 
\put(106.5,75){\line(1,0){7}}
\put(116.5,75){\line(1,0){7}} 
\put(126.5,75){\line(1,0){7}}
\put(136.5,75){\line(1,0){7}} 
\put(125,76.5){\line(0,1){5}}
\put(85,70){\makebox(0,0){$2$}} 
\put(95,70){\makebox(0,0){$3$}}
\put(105,70){\makebox(0,0){$4$}} 
\put(115,70){\makebox(0,0){$5$}}
\put(125,70){\makebox(0,0){$6$}} 
\put(145,70){\makebox(0,0){$2$}}
\put(129,84){\makebox(0,0){$3$}} 
\put(135,70){\makebox(0,0){$4$}}
\put(75,70){\makebox(0,0){$1$}}

%% F_4
\put(-25,50){\makebox(0,0){\large {$\mathbf{F_4}$}}}
\put(0,50){\kr} 
\put(10,50){\kr} 
\put(20,50){\kr}
\put(30,50){\pr} 
\put(40,50){\pr} 
\put(1.5,50){\line(1,0){7}}
\put(11.5,50){\line(1,0){7}} 
\put(21.2,51){\line(1,0){8}}
\put(21.2,49){\line(1,0){8}} 
\put(31.5,50){\line(1,0){7}}
\put(0,45){\makebox(0,0){$0$}} 
\put(30,45){\makebox(0,0){$3$}}
\put(10,45){\makebox(0,0){$1$}} 
\put(40,45){\makebox(0,0){$4$}}
\put(20,45){\makebox(0,0){$2$}}

\put(80,50){\kr} 
\put(90,50){\kr} 
\put(100,50){\kr}
\put(110,50){\pr} 
\put(120,50){\pr} 
\put(81.5,50){\line(1,0){7}}
\put(91.5,50){\line(1,0){7}} 
\put(101,51){\line(1,0){8}}
\put(101,49){\line(1,0){8}} 
\put(111.5,50){\line(1,0){7}}
\put(80,45){\makebox(0,0){$1$}} 
\put(110,45){\makebox(0,0){$4$}}
\put(90,45){\makebox(0,0){$2$}} 
\put(120,45){\makebox(0,0){$2$}}
\put(100,45){\makebox(0,0){$3$}}

%% G_2
\put(-25,25){\makebox(0,0){{\large $\mathbf{G_2}$}}}
\put(0,25){\kr} 
\put(10,25){\kr} 
\put(20,25){\pr}
\put(11.5,25){\line(1,0){7}} 
\put(1.5,25){\line(1,0){7}}
\put(10,26.5){\line(1,0){10}} 
\put(10,23.5){\line(1,0){10}}
\put(0,20){\makebox(0,0){$0$}} 
\put(520,20){\makebox(0,0){$2$}}
\put(10,20){\makebox(0,0){$1$}}

\put(80,25){\kr} 
\put(90,25){\kr} 
\put(100,25){\pr}
\put(91.5,25){\line(1,0){7}} 
\put(81.5,25){\line(1,0){7}}
\put(90,26.5){\line(1,0){10}} 
\put(90,23.5){\line(1,0){10}}
\put(80,20){\makebox(0,0){$1$}} 
\put(100,20){\makebox(0,0){$3$}}
\put(90,20){\makebox(0,0){$2$}}}

\end{picture}}
\caption{The Dynkin numbering of the extended diagrams with marks.}
\label{fig:dynkin+marks}
\end{figure}
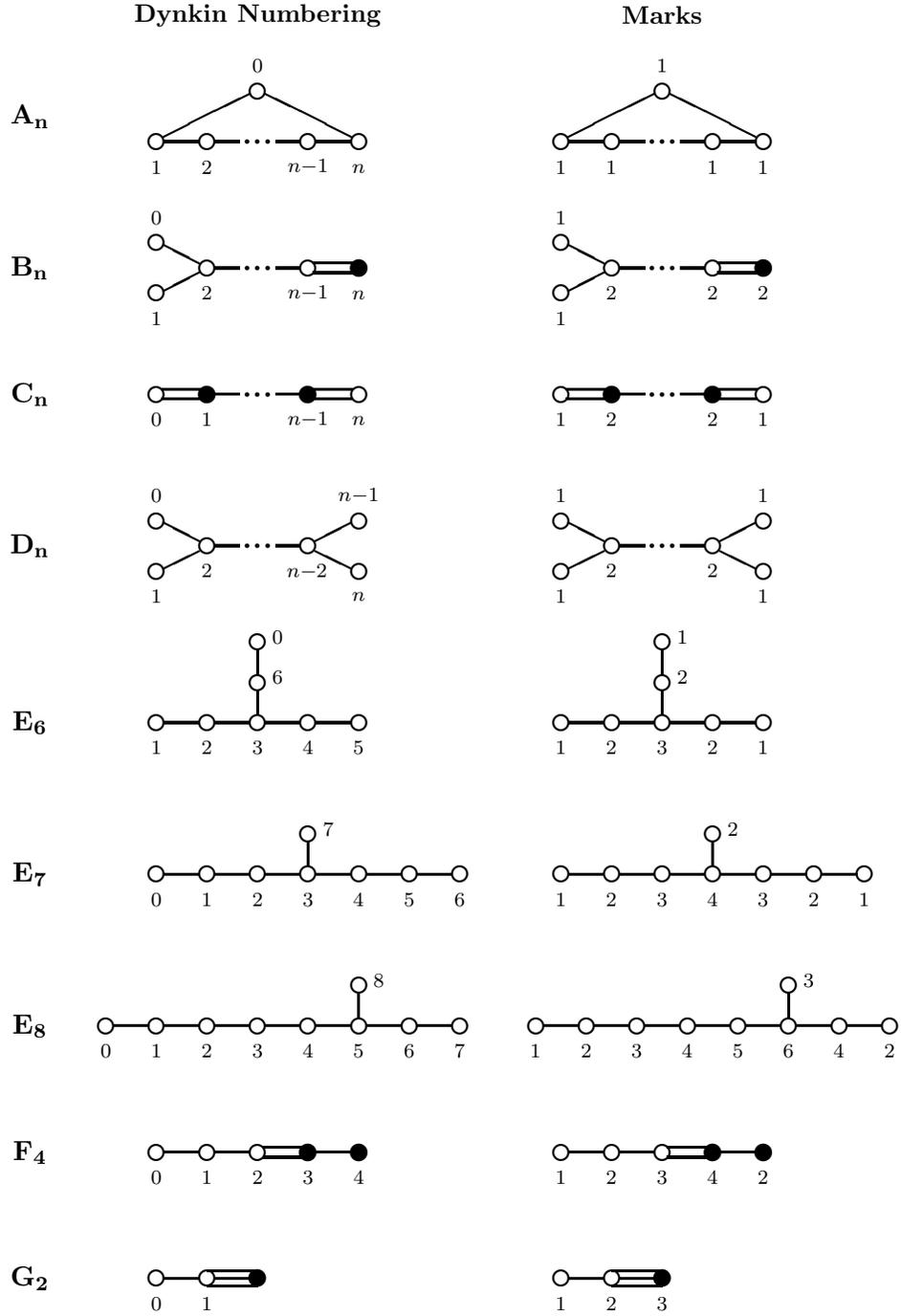

\clearpage

\begin{table}
\parbox{.6\linewidth}{\setlength{\unitlength}{1.95pt}
\def\kr{\circle{3}}
\def\pr{\circle*{3.5}}
\thicklines
\begin{picture}(200,280)
%{\bf \put(75,285){\makebox(0,0){Table 2: Eigenvalues of Central Elements $z=$exp\,$2\pi i\hat\omega_j$ on an}}} 
%{\bf \put(94.5,278){\makebox(0,0){irreducible representation of highest weight $\lambda=\sum_{i=1}^n m_i \omega_i$}}}

\put(-33,270){\line(1,0){175}} 
\put(-33,210){\line(1,0){175}}
\put(-33,150){\line(1,0){175}} 
\put(-33,80){\line(1,0){175}}
\put(-33,3){\line(1,0){175}}

\put(-33,270){\line(0,-1){267}} 
\put(142,270){\line(0,-1){267}}
\put(60,210){\line(0,-1){60}} 
\put(60,80){\line(0,-1){77}}

%% A_n
\put(-25.5,255){\makebox(0,0){\large{$\mathbf{A_n}$}}}
\put(-11.5,245){\makebox(0,0){\bf {Center $\cong\Z_{n+1}$}}}
\put(-3,235){\makebox(0,0){\bf {Generator: exp\,$2\pi i\hat{\omega}_1$}}}

\put(40,240){\circle{3}} 
\put(50,260){\circle{3}}
\put(50,250){\circle{3}} 
\put(50,230){\circle{3}}
\put(50,251.5){\line(0,1){7}} 
\put(50,248.5){\line(0,-1){5}}
\multiput(50,238)(0,2){3}{\circle{0,1}}
\put(50,231.5){\line(0,1){5}} 
\put(50,221.5){\line(0,1){7}}
\put(50,220){\circle{3}} 
\put(40.5,241.5){\line(1,2){8.5}}
\put(40.5,238.5){\line(1,-2){8.5}} {\footnotesize
\put(36.5,240){\makebox(0,0){$1$}}
\put(95,260){\makebox(0,0){$n(m_1+2m_2+\dots +nm_n)$ \qquad \,mod$\,n+1$}} 
\put(95,250){\makebox(0,0){$(n{-}1)(m_1+2m_2+\dots+nm_n)$ \ \ mod$\,n+1$}} \multiput(75,238)(0,2){3}{\circle{0,1}} 
\put(95,230){\makebox(0,0){$2(m_1+2m_2+\dots+nm_n)$ \qquad \ \ mod$\,n+1$}}
\put(95,220){\makebox(0,0){$m_1+2m_2+\dots +nm_n $ \qquad \qquad mod$\,n+1$}}
}

%% B_n
\put(-25,190){\makebox(0,0){{\large$\mathbf{B_n}$}}}
\put(-13.5,180){\makebox(0,0){\bf {Center $\cong\Z_2$}}}
\put(-2,170){\makebox(0,0){\bf {Generator: exp\,$2\pi i\hat{\omega}_1$}}}

\put(25,200){\kr} 
\put(35,200){\kr} 
\put(30,190){\kr}
\put(30,170){\kr} 
\put(30,160){\pr}
\put(29,191.5){\line(-1,2){3.5}} 
\put(31,191.5){\line(1,2){3.5}}
\put(30,171.5){\line(0,1){5}} 
\put(30,188.5){\line(0,-1){5}}
\put(31,161){\line(0,1){8}} 
\put(29,161){\line(0,1){8}}
\multiput(30,178)(0,2){3}{\circle{0,1}} {\footnotesize
\put(21,200){\makebox(0,0){$1$}}
\put(47.5,200){\makebox(0,0){$m_n$\,mod\,$2$}}}

%% C_n
\put(87.5,190){\makebox(0,0){{\large$\mathbf{C_n}$}}}
\put(99,180){\makebox(0,0){\bf {Center $\cong\Z_2$}}}
\put(111,170){\makebox(0,0){\bf {Generator: exp\,$2\pi i\hat{\omega}_n$}}}

\put(75,200){\kr} 
\put(75,190){\pr} 
\put(75,170){\pr}
\put(75,160){\kr} 
\put(76,199){\line(0,-1){8}}
\put(74,199){\line(0,-1){8}} 
\put(75,188.5){\line(0,-1){5}}
\put(75,171.5){\line(0,1){5}} 
\put(74,161.2){\line(0,1){8}}
\put(76,161.2){\line(0,1){8}}
\multiput(75,178)(0,2){3}{\circle{0,1}} {\footnotesize
\put(79,200){\makebox(0,0){$1$}}
\put(105,160){\makebox(0,0){$m_1+m_3+m_5+\dots $\,mod\,$2$}}}

%% D_n
\put(60,140){\makebox(0,0){{\large $\mathbf{D_n}$}}}
\put(8,130){\makebox(0,0){\bf {For $n$ even: \ Center $\cong\Z_2\times \Z_2$}}} 
\put(11.5,120){\makebox(0,0){\bf{Generators: exp\,$2\pi i\hat{\omega}_1$, exp\,$2\pi i\hat{\omega}_n$}}}

\put(105,130){\makebox(0,0){\bf {For $n$ odd: \ Center $\cong\Z_4$}}} \put(101,120){\makebox(0,0){\bf {Generator: exp\,$2\pi i\hat{\omega}_n$}}}

\put(8,102){\kr} 
\put(8,92){\kr} 
\put(18,97){\kr}
\put(38,97){\kr} 
\put(48,102){\kr} 
\put(48,92){\kr}
\put(16.5,98){\line(-2,1){7}} 
\put(16.5,96){\line(-2,-1){7}}
\put(19.5,97){\line(1,0){5}} 
\put(31.5,97){\line(1,0){5}}
\put(39.5,98){\line(2,1){7}} 
\put(39.5,96){\line(2,-1){7}}
\multiput(26,97)(2,0){3}{\circle{0,1}} {\footnotesize
\put(4,102){\makebox(0,0){$1$}}
\put(-12,91){\makebox(0,0){$m_{n-1}+m_n$\,mod\,2}}
\put(94.5,102){\makebox(0,0){$2m_1+2m_3+\dots+nm_{n-1}+(n{-}2)m_n$\,mod\,$4$}}
\put(94.5,91){\makebox(0,0){$2m_1+2m_3+\dots+(n{-}2)m_{n-1}+nm_n$\,mod\,$4$}}}

%% E_6
\put(6.5,50){\makebox(0,0){{\large$\mathbf{E_6}$}}}
\put(18,40){\makebox(0,0){\bf {Center $\cong\Z_3$}}}
\put(30,30){\makebox(0,0){\bf {Generator: exp\,$2\pi i\hat{\omega}_1$}}}

\put(-9,60){\kr} 
\put(-9,50){\kr}  
\put(-9,40){\kr}
\put(-9,30){\kr} 
\put(-9,20){\kr}  
\put(-17,40){\kr}
\put(-25,40){\kr} 
\put(-9,58.5){\line(0,-1){7}}
\put(-9,48.5){\line(0,-1){7}} 
\put(-9,38.5){\line(0,-1){7}}
\put(-9,28.5){\line(0,-1){7}} 
\put(-10.5,40){\line(-1,0){5}}
\put(-18.5,40){\line(-1,0){5}}{\footnotesize
\put(-25,45){\makebox(0,0){$1$}}
\put(22,60){\makebox(0,0){$m_1+2m_2+m_4+2m_5$\,mod\,3}}
\put(22,20){\makebox(0,0){$2m_1+m_2+2m_4+m_5$\,mod\,3}}}

%% E_7
\put(88,50){\makebox(0,0){{\large$\mathbf{E_7}$}}}
\put(99.5,40){\makebox(0,0){\bf {Center $\cong\Z_2$}}}
\put(111,30){\makebox(0,0){\bf {Generator: exp\,$2\pi i\hat{\omega}_6$}}}

\put(75,70){\kr}
\put(75,60){\kr} 
\put(75,50){\kr}
\put(75,40){\kr} 
\put(75,30){\kr} 
\put(75,20){\kr}
\put(75,10){\kr} 
\put(67,40){\kr} 
\put(75,68.5){\line(0,-1){7}}
\put(75,58.5){\line(0,-1){7}} 
\put(75,48.5){\line(0,-1){7}}
\put(75,38.5){\line(0,-1){7}} 
\put(75,28.5){\line(0,-1){7}}
\put(68.5,40){\line(1,0){5}} 
\put(75,18.5){\line(0,-1){7}}
{\footnotesize
\put(98,70){\makebox(0,0){$m_4+m_6+m_7$\,mod\,$2$}}
\put(79,10){\makebox(0,0){$1$}}}

\end{picture}}
\bigskip
\caption{Eigenvalues of central elements $z=$ exp\,$2\pi i\hat\omega_j$ on an irreducible representation of highest weight $\lambda=\sum_{i=1}^n m_i \omega_i$.}
\label{tab:center}
\end{table}
\clearpage

\begin{table}
\begin{center}
{\footnotesize
\begin{tabular}{|c|c|c|}

\hline
$\,\mathbf{L\supset L'}\,$ & $\mathbf{C_{G}(G')}$ &$\mathbf{\frac{C_G(G')}{Z(G)}}$\rule{0pt}{12pt}\\

\hline
\hline
$\,B_2\supset A_1\oplus A_1\,$&  $\,\Z_2\times \Z_2=\langle a,b \rangle$
& $\,\Z_2\,$ \rule{0pt}{12pt}\\
\hline
$\,A_1\,$ & $a:m_1+m_2\,$mod$\,2$&  $a$ \rule{0pt}{12pt}\\
\cline{1-2}
$\,A_1\,$ & $b:m_1\,$mod$\,2$& \rule{0pt}{12pt}\\

\hline
\hline
$B_n\supset B_{n-2}\oplus A_1\oplus A_1$&  $\,\Z_2\times \Z_2=\langle a,b \rangle$
& $\,\Z_2\,$ \rule{0pt}{12pt}\\
\hline
$\,B_{n-2}\,$ & $\,a:m_n\,$mod$\,2$ &  \rule{0pt}{12pt}\\
\cline{1-2}
$\,A_1\,$ & $b:m_{n-2}+m_{n-1}+m_{n}\,$mod$\,2$&  $b$ \rule{0pt}{12pt}\\
\cline{1-2}
$\,A_1\,$ & $a{+}b:m_{n-2}+m_{n-1}\,$mod$\,2$& \rule{0pt}{12pt}\\
\hline
\hline
$B_4\supset A_1\oplus A_3$&  $\,\Z_4=\langle a\rangle$
& $\,\Z_2\,$ \rule{0pt}{12pt}\\
\hline
$\,A_1\,$ & $\,2a:m_4\,$mod$\,2$ &  $a$ \rule{0pt}{12pt}\\
\cline{1-2}
$\,A_3\,$ & $a:2m_1+3m_4\,$mod$\,4$&  \rule{0pt}{12pt}\\
\hline
\hline
$B_n\supset B_{n-3}\oplus A_3$&  $\,\Z_4=\langle a\rangle$
& $\,\Z_2\,$ \rule{0pt}{12pt}\\
\hline
$\,B_{n-3}\,$ & $\,2a:m_n\,$mod$\,2$ &  $a$ \rule{0pt}{12pt}\\
\cline{1-2}
$\,A_3\,$ & $a:2m_{n-4}+2m_{n-3}+3m_{n}\,$mod$\,4$&  \rule{0pt}{12pt}\\
\hline
\hline

$B_n\supset D_{n}$ & ($n$ odd)  \qquad \qquad \quad  \ $\,\Z_4=\langle a\rangle$ \hspace{2.4cm} \
& $\,\Z_2\,$ \rule{0pt}{12pt}\\
\hline
$\,D_n\,$ & $a:2m_1+2m_3+\dots+2m_{n-2}+nm_{n}\,$mod$\,4$&  $a$ \rule{0pt}{12pt}\\
\hline
\hline
$B_n\supset D_{n}$ &  ($n$ even) \ \qquad \quad $\,\Z_2\times\Z_2=\langle a,b\rangle$
\hspace{2cm} \  & $\,\Z_2\,$ \rule{0pt}{12pt}\\
\hline
$\,D_n\,$ & $a:m_1+m_3+\dots+m_{n-1}+\tfrac{n}{2}m_n\,$mod$\,2\,$ & $a$\rule{0pt}{12pt}\\ 
& and $\,b:m_n\,$mod$\,2$&   \rule{0pt}{12pt}\\
\hline
\hline

$B_n\supset D_{n-1}\oplus A_1$ &  ($n$ odd) \qquad  \quad \ \ $\,\Z_2\times\Z_2=\langle a,b\rangle$
 \hspace{2cm} \  & $\,\Z_2\,$ \rule{0pt}{12pt}\\
\hline
$\,D_{n-1}\,$ & $\,a:m_1+m_3+\dots+m_{n-2}+m_{n-1}$ & \rule{0pt}{12pt}\\ 
& $+\tfrac{n-1}{2}m_{n}\,$mod$\,2$ \quad and \quad $b:m_{n}\,$mod$\,2$ &  $a$ \rule{0pt}{12pt}\\
\cline{1-2}
$\,A_1\,$ & $b:m_{n}\,$mod$\,2$&  \rule{0pt}{12pt}\\

\hline
\hline

$B_n\supset D_{n-1}\oplus A_1$ &  ($n$ even) \qquad \qquad \quad $\,\Z_4=\langle a\rangle$
\hspace{2.5cm} & $\,\Z_2\,$ \rule{0pt}{12pt}\\
\hline
$\,D_{n-1}\,$ & $\,a:2m_1+2m_3+\dots+2m_{n-3}+(n-1)m_{n}\,$mod$\,4$ &  $a$ \rule{0pt}{12pt}\\
\cline{1-2}
$\,A_1\,$ & $2a:m_{n}\,$mod$\,2$&  \rule{0pt}{12pt}\\

\hline
\hline
$B_n\supset B_k\oplus D_{n-k}$ & ($n-k$ even) \quad \  $\,\Z_2\times\Z_2=\langle a,b\rangle$
\hspace{2.cm} & $\,\Z_2\,$ \rule{0pt}{12pt}\\
\hline
$\,B_{k}\,$ & $\,a:m_{n}\,$mod$\,2$ &   \rule{0pt}{12pt}\\
\cline{1-2}
$\,D_{n-k}\,$ & $a\,$ and $b:\displaystyle\sum_{i=0}^{(n-k-4)/2}(m_{2k-n+4i+2}+m_{2k-n+4i+3})$ & $b$ \rule{0pt}{12pt}\\
& $+m_{n-2}+m_{n-1}+\tfrac{n-k}{2}m_{n}\,$mod$\,2$ &  \rule{0pt}{12pt}\\
\hline
\end{tabular}}
\end{center}

\bigskip
\caption{Discrete centralizers and relative congruence classes of irreducible representations of the classical simple Lie algebras.}
\label{tab:discrete}
\end{table}
\clearpage

\begin{center}
{\footnotesize
\begin{tabular}{|c|c|c|}
\hline
$\,\mathbf{L\supset L'}\,$ & $\mathbf{C_{G}(G')}$ & $\mathbf{\frac{C_G(G')}{Z(G)}}$\rule{0pt}{12pt}\\
\hline
\hline

$B_n\supset B_k\oplus D_{n-k}$ & ($n-k$ odd) \ \ \ \quad \quad $\,\Z_4=\langle a\rangle$
\hspace{2.5cm} & $\,\Z_2\,$ \rule{0pt}{12pt}\\
\hline
$\,B_{k}\,$ & $\,2a:m_{n}\,$mod$\,2$ &   \rule{0pt}{12pt}\\
\cline{1-2}
$\,D_{n-k}\,$ &$a:\displaystyle\sum_{i=0}^{(n-k-3)/2}2(m_{2k-n+4i+2}+m_{2k-n+4i+3})$ & $a$ \rule{0pt}{12pt}\\
& $+(n-k)m_{n}\,$mod$\,4$ &
\rule{0pt}{12pt}\\

\hline
\hline

$B_n\supset D_k\oplus B_{n-k}$ & ($n,k$ odd) \qquad \qquad   $\,\Z_4=\langle a\rangle$
\hspace{2.5cm} & $\,\Z_2\,$ \rule{0pt}{12pt}\\
\hline
$\,D_{k}\,$ & $a:\displaystyle\sum_{i=0}^{(2k-n-3)/2}2m_{2i+1}$ &  \rule{0pt}{12pt}\\
& $+\displaystyle\sum_{j=0}^{(n-k-2)/2}2(m_{2k-n+4j}+m_{2k-n+4j+1})$ & $a$\rule{0pt}{12pt}\\
& $+km_{n}\,$mod$\,4$& \rule{0pt}{12pt}\\
\cline{1-2}
$\,B_{n-k}\,$ & $\,2a:m_{n}\,$mod$\,2$ &   \rule{0pt}{12pt}\\
\hline
\hline

$B_n\supset D_k\oplus B_{n-k}$ & ($n$ even, $k$ odd) \quad  \
 $\,\Z_4=\langle a\rangle$
\hspace{2.4cm} & $\,\Z_2\,$ \rule{0pt}{12pt}\\
\hline
$\,D_{k}\,$ &  $\,a:\displaystyle\sum_{i=0}^{(2k-n-2)/2}2m_{2i+1}$ & \rule{0pt}{12pt}\\
&$+\displaystyle\sum_{j=0}^{(n-k-3)/2}2(m_{2k-n+4j+2}+m_{2k-n+4j+3})$& $a$ \rule{0pt}{12pt}\\
& $+km_{n}\,$mod$\,4$& \rule{0pt}{12pt}\\
\cline{1-2}
$\,B_{n-k}\,$ & $\,2a:m_{n}\,$mod$\,2$ &
\rule{0pt}{12pt}\\
\hline
\hline

$B_n\supset D_k\oplus B_{n-k}$ & ($n,k$ even) \qquad 
$\,\Z_2\times\Z_2=\langle a, b\rangle$
\hspace{2cm} & $\,\Z_2\,$ \rule{0pt}{12pt}\\
\hline
$\,D_{k}\,$ & $\,a:m_{n}\,$mod$\,2$ \quad and  \quad $\,b: \displaystyle\sum_{i=0}^{(2k-n-2)/2}m_{2i+1}
$ & \rule{0pt}{12pt}\\
& $+\displaystyle\sum_{j=0}^{(n-k-2)/2}(m_{2k-n+4j+2}+m_{2k-n+4j+3})$ & $b$ \rule{0pt}{12pt}\\
& $+\tfrac{k}{2}m_{n}\,$mod$\,2$& \rule{0pt}{12pt}\\
\cline{1-2}
$\,B_{n-k}\,$ & $\,a:m_{n}\,$mod$\,2$ &   \rule{0pt}{12pt}\\

\hline
\hline

$B_n\supset D_k\oplus B_{n-k}$ & ($n$ odd, $k$ even)  
$\,\Z_2\times\Z_2=\langle a, b\rangle$ 
\hspace{2cm} & $\,\Z_2\,$ \rule{0pt}{12pt}\\
\hline
$\,D_{k}\,$ & $\,a:m_{n}\,$mod$\,2$ \quad and  \quad $\,b:\displaystyle\sum_{i=0}^{(2k-n-3)/2}m_{2i+1}$ &  \rule{0pt}{12pt}\\
& $+\displaystyle\sum_{j=0}^{(n-k-1)/2}(m_{2k-n+4j}+m_{2k-n+4j+1})$ & $b$ \rule{0pt}{12pt}\\
& $+\tfrac{k}{2}m_{n}\,$mod$\,2$& \rule{0pt}{12pt}\\
\cline{1-2}
$\,B_{n-k}\,$ & $\,a:m_{n}\,$mod$\,2$ &   \rule{0pt}{12pt}\\
\hline
\hline

$C_n\supset C_{n-1}\oplus A_1$ &  ($n$ even) \quad \quad \ 
$\,\Z_2\times\Z_2=\langle a,b\rangle$ \hspace{2cm}
& $\,\Z_2\,$ \rule{0pt}{12pt}\\
\hline
$\,C_{n-1}\,$ & $\,a:m_1+m_3+\dots+m_{n-1}+m_{n}\,$mod$\,2$ &  $a$ \rule{0pt}{12pt}\\
\cline{1-2}
$\,A_1\,$ & $b:m_{n}\,$mod$\,2$&  \rule{0pt}{12pt}\\
\hline

\end{tabular}}

\bigskip
\sc{Table~\ref{tab:discrete}.} \normalfont (continued)
\end{center}
\clearpage

\begin{center}
{\footnotesize
\begin{tabular}{|c|c|c|}

\hline
$\,\mathbf{L\supset L'}\,$ & $\mathbf{C_{G}(G')}$ & $\mathbf{\frac{C_G(G')}{Z(G)}}$\rule{0pt}{12pt}\\
\hline
\hline
$C_n\supset C_{n-1}\oplus A_1$ & ($n$ odd) \quad \quad \quad \ 
$\,\Z_2\times\Z_2=\langle a,b\rangle$ \hspace{2cm} & $\,\Z_2\,$ \rule{0pt}{12pt}\\
\hline
$\,C_{n-1}\,$ & $\,a:m_1+m_3+\dots+m_{n-2}\,$mod$\,2$ &  $a$ \rule{0pt}{12pt}\\
\cline{1-2}
$\,A_1\,$ & $b:m_{n}\,$mod$\,2$&  \rule{0pt}{12pt}\\
\hline
\hline

$C_n\supset C_{k}\oplus C_{n-k}$ & ($n,k$ odd) \quad \quad \
$\,\Z_2\times\Z_2=\langle a,b\rangle$
 \hspace{1.9cm} & $\,\Z_2\,$ \rule{0pt}{12pt}\\
\hline
$\,C_{k}\,$ & $\,a:\displaystyle\sum_{i=1}^{(2k-n+1)/2} m_{2i-1}$ & \rule{0pt}{12pt}\\
& $+\displaystyle\sum_{j=1}^{{(n-k)}/2}
 (m_{2k-n+4j-1}+m_{2k-n+4j})\,$mod$\,2$ &  $a$ \rule{0pt}{12pt}\\
\cline{1-2}
$\,C_{n-k}\,$ & $b:\displaystyle\sum_{j=1}^{{(n-k)}/2}
 (m_{2k-n+4j-2}+m_{2k-n+4j-1})\,\,$mod$\,2$&  \rule{0pt}{12pt}\\
\hline
\hline

$C_n\supset C_{k}\oplus C_{n-k}$ & ($n,k$ even) \quad \quad
$\,\Z_2\times\Z_2=\langle a,b\rangle$
 \hspace{2.cm} & $\,\Z_2\,$ \rule{0pt}{12pt}\\
\hline
 $\,C_{k}$ & $\,a:\displaystyle\sum_{i=1}^{(2k-n)/2} m_{2i-1}$ & \rule{0pt}{12pt}\\
& $+\displaystyle\sum_{j=1}^{{(n-k)}/2}
 (m_{2k-n+4j-3}+m_{2k-n+4j-2})\,$mod$\,2$ &  $a$ \rule{0pt}{12pt}\\
\cline{1-2}
$\,C_{n-k}\,$ & $b:\displaystyle\sum_{j=1}^{{(n-k)}/2}
 (m_{2k-n+4j-2}+m_{2k-n+4j-1})\,\,$mod$\,2$&  \rule{0pt}{12pt}\\
\hline
\hline

$C_n\supset C_{k}\oplus C_{n-k}$ & ($n$ even, $k$ odd) \ 
$\,\Z_2\times\Z_2=\langle a,b\rangle$
 \hspace{2.cm} & $\,\Z_2\,$ \rule{0pt}{12pt}\\
\hline
$\,C_{k}\,$ & $\,a:\displaystyle\sum_{i=1}^{(2k-n)/2} m_{2i-1}$ & \rule{0pt}{12pt}\\
& $+\displaystyle\sum_{j=1}^{(n-k+1)/2}
 (m_{2k-n+4j-3}+m_{2k-n+4j-2})\,$mod$\,2$ &  $a$ \rule{0pt}{12pt}\\
\cline{1-2}
$\,C_{n-k}\,$ & $b:\displaystyle\sum_{j=1}^{{(n-k-1)}/2}
 (m_{2k-n+4j-2}+m_{2k-n+4j-1})$&  \rule{0pt}{12pt}\\
& $+m_n\,$mod$\,2$&  \rule{0pt}{12pt}\\
\hline
\hline

$C_n\supset C_{k}\oplus C_{n-k}$ & ($n$ odd, $k$ even) \
$\,\Z_2\times\Z_2=\langle a,b\rangle$
 \hspace{2.cm} & $\,\Z_2\,$ \rule{0pt}{12pt}\\
\hline
$\,C_{k}\,$ & $\,a:\displaystyle\sum_{i=1}^{(2k-n+1)/2} m_{2i-1}$ & \rule{0pt}{12pt}\\
& $+\displaystyle\sum_{j=1}^{(n-k-1)/2}
 (m_{2k-n+4j-1}+m_{2k-n+4j})\,$mod$\,2$ &  $a$ \rule{0pt}{12pt}\\
\cline{1-2}
$\,C_{n-k}\,$ & $b:\displaystyle\sum_{j=1}^{{(n-k-1)}/2}
 (m_{2k-n+4j-2}+m_{2k-n+4j-1})$&  \rule{0pt}{12pt}\\
& $+m_n\,$mod$\,2$&  \rule{0pt}{12pt}\\

\hline

\end{tabular}}

\bigskip
\sc{Table~\ref{tab:discrete}.} \normalfont (continued)
\end{center}

\clearpage

\begin{center}
{\footnotesize
\begin{tabular}{|c|c|c|}

\hline
$\,\mathbf{L\supset L'}\,$ & $\mathbf{C_{G}(G')}$ & $\mathbf{\frac{C_G(G')}{Z(G)}}$\rule{0pt}{12pt}\\
\hline
\hline
$D_4\supset A_1{\oplus}A_1{\oplus}A_1{\oplus}A_1$ & $\,\Z_2\times\Z_2\times\Z_2=\langle a,b,c\rangle$ \
& $\Z_2$ \rule{0pt}{12pt}\\
\hline
$A_1$ & $a:m_2+m_3\,$mod$\,2\,$ &  \rule{0pt}{12pt}\\
\cline{1-2}
$A_1$ & $b:m_2+m_4\,$mod$\,2$& $b$ \rule{0pt}{12pt}\\
\cline{1-2}
$\,A_1\,$ & $c:m_1+m_2+m_3+m_4\,$mod$\,2$&   \rule{0pt}{12pt}\\
\cline{1-2}
$\,A_1\,$ & $a{+}b{+}c:m_1+m_2\,$mod$\,2$&  \rule{0pt}{12pt}\\

\hline
\hline
$D_n\supset D_{n-2}\oplus A_1\oplus A_1$ & ($n$ even) \quad \
$\,\Z_2\times\Z_2\times\Z_2=\langle a,b,c\rangle$ \hspace{1.5cm}
& $\Z_2$ \rule{0pt}{12pt}\\
\hline
$\,D_{n-2}\,$ & $a:m_1+m_3+\dots+m_{n-5}+m_{n-2}+\tfrac{n}{2}m_{n-1}$ &  \rule{0pt}{12pt}\\
&$+(1+\tfrac{n}{2})m_n\,$mod$\,2\,$ & $a$ \rule{0pt}{12pt}\\
& and  \quad $b:m_{n-1}+m_n\,$mod$\,2$& \rule{0pt}{12pt}\\
\cline{1-2}
$\,A_1\,$ & $c:m_{n-3}+m_{n-2}+m_{n-1}+m_n\,$mod$\,2$&   \rule{0pt}{12pt}\\
\cline{1-2}
$\,A_1\,$ & $b{+}c:m_{n-3}+m_{n-2}\,$mod$\,2$&  \rule{0pt}{12pt}\\

\hline
\hline
$D_n\supset D_{n-2}\oplus A_1\oplus A_1$ & ($n$ odd) \qquad  \quad
$\,\Z_4\times\Z_2=\langle a,b\rangle$ \hspace{2.cm}
& $\,\Z_2\,$ \rule{0pt}{12pt}\\
\hline
$\,D_{n-2}\,$ & $a:2m_1+2m_3+\dots+2m_{n-4}+2m_{n-3}$ &  \rule{0pt}{12pt}\\
&$+nm_{n-1}+(n-2)m_n\,$mod$\,4\,$& $a$
\rule{0pt}{12pt}\\
\cline{1-2}
$\,A_1\,$ & $b:m_{n-3}+m_{n-2}+m_{n-1}+m_n\,$mod$\,2$&   \rule{0pt}{12pt}\\
\cline{1-2}
$\,A_1\,$ & $2a{+}b:m_{n-3}+m_{n-2}\,$mod$\,2$&  \rule{0pt}{12pt}\\

\hline
\hline

$D_6\supset A_3\oplus A_3$ & $\,\Z_4\times\Z_2=\langle a,a{+}b\rangle$
& $\Z_2$ \rule{0pt}{12pt}\\
\hline
$A_3$ & $a:2m_2+2m_3+m_5+3m_6\,$mod$\,4$& $a$
\rule{0pt}{9pt}\\
\cline{1-2}
$A_3$ & $b:2m_1+2m_2+3m_5+3m_6\,$mod$\,4$&  \rule{0pt}{12pt}\\

\hline
\hline

$D_n\supset D_{n-3}\oplus A_3$ & ($n$ even) \quad  \quad
$\,\Z_4\times\Z_2=\langle a,a{+}b\rangle$ \hspace{1.8cm}
& $\Z_2$ \rule{0pt}{12pt}\\
\hline
$\,D_{n-3}\,$ & $a:2m_1+2m_3+\dots + 2m_{n-7}+2m_{n-4}+$ & \rule{0pt}{9pt}\\
& $2m_{n-3}+(n-1)m_{n-1}+(n+1)m_n\,$mod$\,4$& $a$
\rule{0pt}{9pt}\\
\cline{1-2}
$\,A_3\,$ & $b:2m_{n-5}+2m_{n-4}+3m_{n-1}+3m_n\,$mod$\,4$&  \rule{0pt}{12pt}\\

\hline
\hline

$D_n\supset D_{n-3}\oplus A_3$ & ($n$ odd) \quad \ \ \quad
$\,\Z_2\times\Z_4=\langle a,b\rangle$ \hspace{1.9cm}
& $\Z_2$ \rule{0pt}{12pt}\\
\hline
$\,D_{n-3}\,$ & $a:m_1+m_3+\dots+m_{n-6}+m_{n-5}+m_{n-2}$ & \rule{0pt}{9pt}\\
& $+\tfrac{n-1}{2} m_{n-1}+\tfrac{n+1}{2}m_n\,$mod$\,2\,$ & $a$\rule{0pt}{9pt}\\
& and  \quad $2b:m_{n-1}+m_n\,$mod$\,2$&
 \rule{0pt}{9pt}\\
\cline{1-2}
$\,A_3\,$ & $b:2m_{n-5}+2m_{n-4}+3m_{n-1}+3m_n\,$mod$\,4$&  \rule{0pt}{12pt}\\

\hline

\end{tabular}}

\bigskip
\sc{Table~\ref{tab:discrete}.} \normalfont (continued)
\end{center}

\clearpage

\begin{center}
{\footnotesize
\begin{tabular}{|c|c|c|}

\hline
$\,\mathbf{L\supset L'}\,$ & $\mathbf{C_{G}(G')}$ & $\mathbf{\frac{C_G(G')}{Z(G)}}$\rule{0pt}{12pt}\\
\hline
\hline
$D_n\supset D_{k}\oplus D_{n-k}$ & ($n,k$ odd) \quad  \qquad
$\,\Z_4  \times\Z_2=\langle a,b\rangle$ \hspace{2.1cm}
 \qquad  & $\,\Z_2\,$ \rule{0pt}{12pt}\\
\hline
$\,D_{k}\,$ & $\,a:\displaystyle\sum_{i=1}^{(2k-n+1)/2} 2m_{2i-1}$ &\rule{0pt}{12pt}\\
& $+\displaystyle\sum_{j=1}^{(n-k-2)/2} 2(m_{2k-n+4j-1}+m_{2k-n+4j})$ &\rule{0pt}{12pt}\\
& $+k(m_{n-1}+m_n)\,$mod$\,4$ & $a$  \rule{0pt}{12pt}\\
\cline{1-2}
$\,D_{n-k}\,$ & $2a:m_{n-1}+m_n\,$mod$\,2$ \quad and &\rule{0pt}{12pt}\\
& $b:\displaystyle\sum_{j=1}^{(n-k-2)/2}
 (m_{2k-n+4j-2}+m_{2k-n+4j-1})$&  \rule{0pt}{12pt}\\
& $+m_{n-2}+\tfrac{n-k-2}{2}m_{n-1}+\tfrac{n-k}{2}m_n\,$mod$\,2\,\,$ & \rule{0pt}{12pt}\\
\hline
\hline

$D_n\supset D_{k}\oplus D_{n-k}$ & ($n,k$ even) \qquad
$\,\Z_2\times\Z_2 \times\Z_2=\langle a,b,c\rangle$
 \hspace{1.7cm} & $\Z_2$ \rule{0pt}{12pt}\\
\hline
$\,D_{k}\,$ & $\,a:m_{n-1}+m_{n}\,$mod$\,2$  \quad and \quad
$\,b:\displaystyle\sum_{i=1}^{(2k-n)/2} m_{2i-1}$ &  \rule{0pt}{12pt}\\
& $+\displaystyle\sum_{j=1}^{(n-k)/2}
 (m_{2k-n+4j-3}+m_{2k-n+4j-2})$ &\rule{0pt}{12pt}\\
& $+\tfrac{k}{2}(m_{n-1}+m_n)\,$mod$\,2$ &  $b$\rule{0pt}{12pt}\\
\cline{1-2}
$\,D_{n-k}\,$ & $a:m_{n-1}+m_n\,$mod$\,2$ \quad and &\rule{0pt}{12pt}\\
& $c:\displaystyle\sum_{j=1}^{(n-k-2)/2}
 (m_{2k-n+4j-2}+m_{2k-n+4j-1})$&  \rule{0pt}{12pt}\\
& $+m_{n-2}+\tfrac{n-k-2}{2}m_{n-1}+\tfrac{n-k}{2}m_n\,$mod$\,2\,\,$ & \rule{0pt}{12pt}\\
\hline
\hline

$D_n\supset D_{k}\oplus D_{n-k}$ & ($n$ even, $k$ odd) \
$\,\Z_4  \times\Z_2=\langle a,a{+}b\rangle$
 \hspace{2.2cm} & $\Z_2$ \rule{0pt}{12pt}\\
\hline
$\,D_{k}\,$ & $a:\displaystyle\sum_{i=1}^{(2k-n)/2} 2m_{2i-1} $ &\rule{0pt}{12pt}\\ 
& $+\displaystyle\sum_{j=1}^{(n-k-1)/2}
 2(m_{2k-n+4j-3}+m_{2k-n+4j-2})$ &$a$\rule{0pt}{12pt}\\
& $+k(m_{n-1}+m_n)\,$mod$\,4$  &   \rule{0pt}{12pt}\\
\cline{1-2}
$\,D_{n-k}\,$ & $b:\displaystyle\sum_{j=1}^{(n-k-1)/2}
 2(m_{2k-n+4j-2}+m_{2k-n+4j-1})$& \rule{0pt}{12pt}\\
& $+(n{-}k{-}2)m_{n-1}+(n{-}k)m_{n}\,$mod$\,4$ &  \rule{0pt}{12pt}\\

\hline
\end{tabular}}

\bigskip
\sc{Table~\ref{tab:discrete}.} \normalfont (continued)
\end{center}

\clearpage

\begin{center}
{\footnotesize
\begin{tabular}{|c|c|c|}

\hline
$\,\mathbf{L\supset L'}\,$ & $\mathbf{C_{G}(G')}$ & $\mathbf{\frac{C_G(G')}{Z(G)}}$\rule{0pt}{12pt}\\
\hline
\hline

$D_n\supset D_{k}\oplus D_{n-k}$ & ($n$ odd, $k$ even) \ \
$\,\Z_2  \times\Z_4=\langle a,b\rangle$
 \hspace{2.2cm} & $\Z_2$ \rule{0pt}{12pt}\\
\hline
$\,D_{k}\,$ &  $a:\displaystyle\sum_{i=1}^{(2k-n+1)/2} m_{2i-1}$ &\rule{0pt}{12pt}\\
& $+\displaystyle\sum_{j=1}^{(n-k-1)/2}
 (m_{2k-n+4j-1}+m_{2k-n+4j})$ &$a$\rule{0pt}{12pt}\\
& $+\tfrac{k}{2}(m_{n-1}+m_n)\,$mod$\,2$ \quad and\quad $2b:m_{n-1}+m_{n}\,$mod$\,2$ &   \rule{0pt}{12pt}\\
\cline{1-2}
$\,D_{n-k}\,$ & $b:\displaystyle\sum_{j=1}^{(n-k-1)/2}
 2(m_{2k-n+4j-2}+m_{2k-n+4j-1})$& \rule{0pt}{12pt}\\
& $+(n{-}k{-}2)m_{n-1}+(n{-}k)m_{n}\,$mod$\,4$ &  \rule{0pt}{12pt}\\

\hline
\end{tabular}}

\bigskip
\sc{Table~\ref{tab:discrete}.} \normalfont (continued)
\end{center}

\clearpage

\begin{table}
\begin{center}
{\footnotesize
\begin{tabular}{|c|c|c|}

\hline
$\,\mathbf{L\supset L'}\,$ & $\mathbf{C_{G}(G')}$ & $\mathbf{\frac{C_G(G')}{Z(G)}}$\rule{0pt}{12pt}\\

\hline
\hline
$\,E_6\supset A_5\oplus A_1\,$&  $\,\Z_6=\langle a \rangle$
& $\,\Z_2\,$ \rule{0pt}{12pt}\\
\hline
$\,A_5\,$ &  $a:4m_1+5m_2+3m_3+m_4+2m_5+3m_6\,$mod$\,6$  & $3a$\rule{0pt}{12pt}\\
\cline{1-2}
$\,A_1\,$ & $\,3a:m_2+m_3+m_4+m_6\,$mod$\,2$ &  \rule{0pt}{12pt}\\

\hline
\hline

$E_6\supset A_2\oplus A_2\oplus A_2$&
$\,\Z_3\times\Z_3=\langle a,b\rangle\,$ &
$\,\Z_3\,$ \rule{0pt}{12pt}\\
\hline
$A_2$& $a:m_1+m_5+m_6\,$mod$\,3$ & \rule{0pt}{12pt}\\
\cline{1-2}
$A_2$& $b:2m_2+m_4+m_5+2m_6\,$mod$\,3$ & $b$ \rule{0pt}{12pt}\\
\cline{1-2}
$A_2$& $2a{+}b:2m_1+2m_2+m_4+m_6\,$mod$\,3$ & \rule{0pt}{12pt}\\

\hline
\hline

$\,E_7\supset A_5\oplus A_2\,$&  $\,\Z_6=\langle a\rangle\,$ &
$\,\Z_3\,$ \rule{0pt}{12pt}\\
\hline
$\,A_5\,$ & $\,a:2m_2+m_4+4m_5+3m_6+5m_7\,$mod$\,6$ & $a$\rule{0pt}{12pt}\\
\cline{1-2}
$\,A_2\,$ & $\,2a:2m_2+m_4+m_5+2m_7\,{\rm mod}\,3\,$&   \rule{0pt}{12pt}\\

\hline
\hline 

$\,E_7\supset A_7\,$ &
$\,\Z_4=\langle a\rangle$& $\,\Z_2\,$ \rule{0pt}{12pt}\\
\hline
$\,A_7\,$& $a:2m_1+2m_2+m_4+m_6+3m_7\,{\rm mod}\,4\,$ & $a$ \rule{0pt}{12pt}\\

\hline
\hline

$\,E_7\supset D_6\oplus A_1\,$ &  $\,\Z_2\times\Z_2=\langle a,b\rangle\,$ &
$\,\Z_2\,$\rule{0pt}{12pt}\\
\hline
$\,D_6\,$ & $\,a:m_2+m_3+m_6\,$mod$\,2$ \quad and &  \rule{0pt}{12pt}\\
& $b:m_4+m_6+m_7\,$mod$\,2$ & $b$\rule{0pt}{12pt}\\
\cline{1-2}
$\,A_1\,$& $\,a{+}b:m_2+m_3+m_4+m_7\,$mod$\,2$ & \rule{0pt}{12pt}\\

\hline
\hline

$\,E_8\supset D_8\,$ &
$\,\Z_2=\langle a\rangle$ & $\Z_2$ \rule{0pt}{12pt}\\
\hline
$\,D_8\,$ & $\,a:m_1+m_2+m_5+m_6\,{\rm mod}\,2\,$ (and $\,0\,$mod$\,2$)&$a$
 \rule{0pt}{12pt}\\

\hline
\hline

$\,E_8\supset A_8\,$ &
$\,\Z_3=\langle a\rangle$ & $\Z_3$ \rule{0pt}{12pt}\\
\hline
$A_8$&$\,a:m_1+m_4+2m_6+2m_7+m_8\,{\rm mod}\,3\,$ & $a$ \rule{0pt}{12pt}\\

\hline
\hline

$\,E_8\supset E_7\oplus A_1\,$ &
$\,\Z_2=\langle a\rangle$ & $\Z_2$ \rule{0pt}{12pt}\\
\hline
$E_7$&$\,a:m_5+m_8\,{\rm mod}\,2\,$ & $a$  \rule{0pt}{12pt}\\
\cline{1-2}
$A_1$&$\,a:m_5+m_8\,{\rm mod}\,2\,$ &   \rule{0pt}{12pt}\\

\hline
\hline 

$\,E_8\supset A_4\oplus A_4\,$ &
$\,\Z_5=\langle a\rangle$ & $\Z_5$ \rule{0pt}{12pt}\\
\hline
$\,A_4\,$&
$\,a:m_1+3m_2+m_3+4m_6+4m_7+2m_8\,{\rm mod}\,5\,$ & $a$ \rule{0pt}{12pt}\\
\cline{1-2}
$\,A_4\,$&
$3a:3m_1+4m_2+3m_3+2m_6+2m_7+m_8\,{\rm mod}\,5$ & \rule{0pt}{12pt}\\

\hline
\hline

$\,E_8\supset E_6\oplus A_2\,$ &
$\,\Z_3=\langle a\rangle$ & $\Z_3$ \rule{0pt}{12pt}\\
\hline
$\,E_6\,$& $\,a:m_3+2m_7+m_8\,{\rm mod}\,3\,$ &$a$\rule{0pt}{12pt}\\
\cline{1-2}
$\,A_2\,$& $\,a:m_3+2m_7+m_8\,$mod$\,3$& \rule{0pt}{12pt}\\

\hline
\hline

$\,G_2\supset A_1\oplus A_1\,$ &
$\,\Z_2=\langle a\rangle$ & $\Z_2$ \rule{0pt}{12pt}\\
\hline
$\,A_1\,$& $\,a:m_1+m_2\,{\rm
mod}\,2\,$& $a$
\rule{0pt}{12pt} \\
\cline{1-2}
$\,A_1\,$& $\,a:m_1+m_2\,{\rm
mod}\,2\,$ &
\rule{0pt}{12pt}\\
\hline
\end{tabular}}
\end{center}

\bigskip
\caption{Discrete centralizers and relative congruence classes of irreducible representations of the exceptional simple Lie algebras.}
\label{tab:discrete.exc}
\end{table}

\clearpage

\begin{center}
{\footnotesize
\begin{tabular}{|c|c|c|}
\hline
$\,\mathbf{L\supset L'}\,$ & $\mathbf{C_{G}(G')}$ & $\mathbf{\frac{C_G(G')}{Z(G)}}$\rule{0pt}{12pt}\\
\hline
\hline 

$\,G_2\supset A_2\,$& \ \ \quad \quad \qquad \qquad \qquad
$\,\Z_3=\langle a\rangle$ \ \ \quad \quad \qquad \qquad \qquad \qquad & $\Z_3$ \rule{0pt}{12pt}\\
\hline
$\,A_2\,$& $\,a:m_2\,{\rm mod}\,3\,$  & $a$ \rule{0pt}{12pt}\\

\hline
\hline 

\hspace{.25cm} $\,F_4\supset A_2\oplus A_2\,$ \hspace{.25cm} &
$\,\Z_3=\langle a\rangle$ & $\Z_3$ \rule{0pt}{12pt}\\
\hline
$\,A_2\,$ &
$\,a:m_2+m_4\,{\rm mod}\,3\,$ & $a$ \rule{0pt}{12pt}\\
\cline{1-2}
$\,A_2\,$ &
$\,a:m_2+m_4\,{\rm mod}\,3\,$ & \rule{0pt}{12pt}\\

\hline
\hline

$\,F_4\supset B_4\,$ &
$\,\Z_2=\langle a\rangle$ & $\Z_2$ \rule{0pt}{12pt}\\
\hline
$\,B_4\,$& $\,a:m_3+m_4\,{\rm mod}\,2\,$ & $a$  \rule{0pt}{12pt}\\

\hline
\hline 

$\,F_4\supset C_3\oplus A_1\,$ &
$\,\Z_2=\langle a\rangle$ & $\Z_2$ \rule{0pt}{12pt}\\
\hline
$\,C_3\,$ & $\,a:m_2+m_4\,{\rm mod}\,2\,$ & $a$  \rule{0pt}{12pt}\\
\cline{1-2}
$\,A_1\,$ & $\,a:m_2+m_4\,{\rm mod}\,2\,$ & \rule{0pt}{12pt}\\

\hline
\end{tabular}}

\bigskip
\sc{Table~\ref{tab:discrete.exc}.} \normalfont (continued)
\end{center}

\clearpage

\begin{table}
\begin{center}
{\footnotesize
\begin{tabular}{|c|c|}

\hline

$\mathbf{A_n\supset A_k\oplus A_{n-k-1} \oplus H_1}$&  $C_G(G')\cong U_1\times Z(A_n)$
 \rule{0pt}{12pt}\\
& \hspace{2.5cm} $\,C_G(G')/U_1\cong \Z_d\,,\ d=gcd(k{+}1,n{+}1)$
 \rule{0pt}{12pt}\\
\hline
\hline
$Z(A_n)$ & $m_1+2m_2+\dots+nm_n \ \ $mod$\,(n{+}1)$   \rule{0pt}{12pt}\\
\hline
$Z(A_k)$ & $m_1+2m_2+\dots+km_k \ \ $mod$\,(k{+}1)$   \rule{0pt}{12pt}\\
\hline
$Z(A_{n-k-1})$ & $m_{k+2}+2m_{k+3}+\dots+(n{-}k{-}1)m_n \ \ $mod$\,(n{-}k)$ \rule{0pt}{12pt}\\
\hline
$\,H_1\,$ & $\C \frac{1}{n+1}\left( (n{-}k)\displaystyle\sum_{i=1}^{k+1}i\hat\alpha_i + (k{+}1) \displaystyle\sum_{i=1}^{n-k-1}(n{-}k{-}i)\hat\alpha_{k+1+i} \right)$ \rule{0pt}{12pt}\\
\hline
Relative congruence relation& $(n{-}k)\displaystyle\sum_{i=1}^{k+1}im_i + (k{+}1) \displaystyle\sum_{i=1}^{n-k-1}(n{-}k{-}i)m_{k+1+i}$ \rule{0pt}{12pt}\\

\hline
\hline

$\mathbf{B_n\supset B_{n-1}\oplus H_1}$ & $C_G(G')\cong U_1$
\rule{0pt}{12pt}\\
\hline\hline
$\,Z(B_n)\,$ & $m_{n}\,$mod$\,2$    \rule{0pt}{12pt}\\
\hline
$\,Z(B_{n-1})\,$ & $m_{n}\,$mod$\,2$    \rule{0pt}{12pt}\\
\hline
$\,H_1\,$ &$\C \frac{1}{2}\left( \hat\alpha_n \right)$ 
\rule{0pt}{12pt}\\
\hline
Relative congruence relation& $m_n$ \rule{0pt}{12pt}\\

\hline
\hline

$\mathbf{C_n\supset A_{n-1}\oplus H_1}$ & $C_G(G')\cong U_1$ \rule{0pt}{12pt}\\
\hline
\hline
$\,Z(C_{n})\,$ & $m_1+m_3+\dots+m_{2[\frac{n+1}{2}]-1}\,$mod$\,2$    \rule{0pt}{12pt}\\
\hline
$\,Z(A_{n-1})\,$ & $m_1+m_3+\dots+m_{2[\frac{n+1}{2}]-1}\,$mod$\,n$    \rule{0pt}{12pt}\\
\hline
$\,H_1\,$ & $\C \frac{1}{2}\left( \hat\alpha_1+\hat\alpha_3+\dots+\hat\alpha_{2[\frac{n+1}{2}]-1}\right)$ \rule{0pt}{12pt}\\
\hline
Relative congruence relation& $m_1+m_3+\dots+m_{2[\frac{n+1}{2}]-1}$ \rule{0pt}{12pt}\\
\hline
\hline

$\mathbf{D_n\supset A_{n-1}\oplus H_1}$ & ($n$ even) \ \qquad \qquad   $C_G(G')\cong U_1\times\Z_2$ \hspace{2.4cm}
\rule{0pt}{12pt}\\
\hline
\hline
$\,Z(D_{n})\,$ & $m_{n-1}+m_n\,$mod$\,2$ \quad and    \rule{0pt}{12pt}\\
&$m_1+m_3+\dots+m_{n-3}+(1{+}\tfrac{n}{2}) m_{n-1}+\tfrac{n}{2}m_n\,$mod$\,2$    \rule{0pt}{12pt}\\
\hline
$\,Z(A_{n-1})\,$ & $m_1+m_3+\dots+m_{n-3}+(1{+}\tfrac{n}{2}) m_{n-1}+\tfrac{n}{2}m_n\,$mod$\,n$    \rule{0pt}{12pt}\\
\hline
$\,H_1\,$ &  $\C \frac{1}{2}\left( \hat\alpha_1+\hat\alpha_3+\dots+\hat\alpha_{n-1}\right)$ \rule{0pt}{12pt}\\
\hline
Relative congruence relation& $m_1+m_3+\dots+m_{n-1}$ \rule{0pt}{12pt}\\
\hline
\hline

$\mathbf{D_n\supset A_{n-1}\oplus H_1}$ & ($n$ odd)\qquad \qquad \qquad$C_G(G')\cong U_1$ \hspace{2.8cm}
 \rule{0pt}{12pt}\\
\hline
\hline
$\,Z(D_{n})\,$ & $2m_1+2m_3+\dots+2m_{n-2}+(n{-}2)m_{n-1}+nm_n\,$mod$\,4$    \rule{0pt}{12pt}\\
\hline
$\,Z(A_{n-1})\,$ & $m_1+m_3+\dots+m_{n-2}+\tfrac{n-1}{2} m_{n-1}+\tfrac{n+1}{2}m_n\,$mod$\,n$     \rule{0pt}{12pt}\\
\hline
$\,H_1\,$ &$\C \frac{1}{4}\left(2\hat\alpha_1+2\hat\alpha_3+\dots+2\hat\alpha_{n-2}-\hat\alpha_{n-1}+\hat\alpha_n\right)$  \rule{0pt}{12pt}\\
\hline
Relative congruence relation& $2m_1+2m_3+\dots+2m_{n-2}-m_{n-1}+m_n$ \rule{0pt}{12pt}\\
\hline
\end{tabular}}
\end{center}

\bigskip
\caption{Continuous centralizers and relative congruence classes of irreducible representations of the simple Lie algebras.}
\label{tab:continuous}
\end{table}

\clearpage

\begin{center}
{\footnotesize
\begin{tabular}{|c|c|}

\hline
$\mathbf{D_n\supset D_{n-1}\oplus H_1}$ & ($n$ even) \ \qquad \qquad   $C_G(G')\cong U_1\times\Z_2$ \hspace{2.4cm}
 \rule{0pt}{12pt}\\
\hline
\hline
$Z(D_n)$ & $m_{n-1}+m_n\,$mod$\,2$ \quad and    \rule{0pt}{12pt}\\
&$m_1+m_3+\dots+m_{n-3}+(1{+}\tfrac{n}{2}) m_{n-1}+\tfrac{n}{2}m_n\,$mod$\,2$    \rule{0pt}{12pt}\\
\hline
$Z(D_{n-1})$ & $2m_1+2m_3+\dots+2m_{n-3}+(n{-}1)(m_{n-1}+m_n)\,$mod$\,4$    \rule{0pt}{12pt}\\
\hline
$\,H_1\,$ & $\C \frac{1}{2}\left(\hat\alpha_{n-1}-\hat\alpha_n\right)$
\rule{0pt}{12pt}\\
\hline
Relative congruence relation& $m_{n-1}-m_n$ \rule{0pt}{12pt}\\
\hline
\hline
$\mathbf{D_n\supset D_{n-1}\oplus H_1}$ & ($n$ odd)  \qquad \qquad  $C_G(G')\cong U_1\times Z(D_n)$ \hspace{2.1cm}
 \rule{0pt}{12pt}\\
& $C_G(G')/U_1\cong \Z_2$ \rule{0pt}{12pt}\\
\hline
\hline
$Z(D_n)$ & $2m_1+2m_3+\dots+2m_{n-2}+(n{-}2)m_{n-1}+nm_n\,$mod$\,4$    \rule{0pt}{12pt}\\
\hline
$Z(D_{n-1})$ & $m_{n-1}+m_n\,$mod$\,2$ \quad and     \rule{0pt}{12pt}\\
&$m_1+m_3+\dots+m_{n-2}+\tfrac{n-1}{2}(m_{n-1}+m_n)\,$mod$\,2$    \rule{0pt}{12pt}\\
\hline
$\,H_1\,$ & $\C \frac{1}{2}\left(\hat\alpha_{n-1}-\hat\alpha_n\right)$
\rule{0pt}{12pt}\\
\hline
Relative congruence relation& $m_{n-1}-m_n$ \rule{0pt}{12pt}\\
\hline
\hline
$\mathbf{E_6\supset D_5\oplus H_1}$ & $C_G(G')\cong U_1$
 \rule{0pt}{12pt}\\
\hline
\hline
$Z(E_6)$ & $m_1-m_2+m_4-m_5\,$mod$\,3$    \rule{0pt}{12pt}\\
\hline
$Z(D_5)$ & $m_1-m_2+m_4-m_5\,$mod$\,4$    \rule{0pt}{12pt}\\
\hline
$\,H_1\,$ & $\C \frac{1}{3}\left(\hat\alpha_1-\hat\alpha_2+\hat\alpha_4-\hat\alpha_5\right)$   \rule{0pt}{12pt}\\
\hline
Relative congruence relation& $m_1-m_2+m_4-m_5$ \rule{0pt}{12pt}\\
\hline
\hline
\hspace{.75cm} $\mathbf{E_7\supset E_6\oplus H_1}$ \hspace{.75cm} & $C_G(G')\cong U_1$
\rule{0pt}{12pt}\\
\hline
\hline
$Z(E_7)$ & $m_4+m_6+m_7\,$mod$\,2$    \rule{0pt}{12pt}\\
\hline
$Z(E_6)$ & \hspace{2.3cm} $m_4+m_6+m_7\,$mod$\,3$ \hspace{2.3cm}    \rule{0pt}{12pt}\\
\hline
$\,H_1\,$ & $\C \frac{1}{2}\left(\hat\alpha_4+\hat\alpha_6+\hat\alpha_7\right)$ 
\rule{0pt}{12pt}\\

\hline
Relative congruence relation& $m_4+m_6+m_7$ \rule{0pt}{12pt}\\
\hline
\end{tabular}}

\bigskip
\sc{Table~\ref{tab:continuous}.} \normalfont (continued)
\end{center}

\end{document}